\documentclass[11pt]{article}
\usepackage{jheppub}

\usepackage{tikz-cd}
\tikzset{commutative diagrams/arrow style=math font}
\usepackage[multiple]{footmisc}

\usepackage{amssymb}
\usepackage{amsfonts}
\usepackage{amsbsy}
\usepackage{amsmath}
\usepackage{amsthm}
\usepackage{mathrsfs}
\usepackage{graphicx}
\usepackage{epstopdf}
\usepackage[vcentermath]{youngtab}
\usepackage{multirow}
\usepackage{latexsym}
\usepackage{array}
\def\M{\mathcal{M}}

\def\n3a{t}

\newcommand\field[1]{{\ensuremath{\mathbb{{#1}}}}}

\newcommand{\MM}{\mathcal{M}}

\newcommand{\p}{\partial}

\newcommand{\HH}{\mathcal{H}}

\def\ov{\over}

\def\vev#1{\langle#1\rangle}

\def\eq#1{(\ref{#1})}

\def\Om{{\Omega}}
\def \th{{\theta}}
\def\a{{\alpha}}
\def\b{{\beta}}

\def \om {\omega}
\def \ra {\rightarrow}

\def\sig{{\sigma}}

\def\l0{{\lambda_o}}

\def\bz{{\bar{z}}}

\def\l{{\lambda}}
\def\lb{{\bar\lambda}}
\def\inta{{\int_\alpha}}
\def\intb{{\int_\beta}}
\def\intal{{\int_\alpha \l}}
\def\intbl{{\int_\beta \l}}
\def\intalb{{\int_\alpha \lb}}
\def\intblb{{\int_\beta \lb}}

\def\BB{{\cal B}}
\def\FF{{\cal F}}

\newcommand{\be}{\begin{equation}}
\newcommand{\ee}{\end{equation}}
\newcommand{\bea}{\begin{eqnarray}}
\newcommand{\eea}{\end{eqnarray}}

\newcommand{\bln}{\begin{align}}
\newcommand{\eln}{\end{align}}
\newcommand{\bst}{\begin{split}}
\newcommand{\est}{\end{split}}
\newcommand{\bi}{\begin{itemize}}
\newcommand{\ei}{\end{itemize}}
\newcommand{\ben}{\begin{enumerate}}
\newcommand{\een}{\end{enumerate}}

\newtheorem*{theorem*}{Theorem}

\newtheorem*{lemma*}{Lemma}
\newtheorem{corollary}[equation]{Corollary}
\newtheorem{proposition}[equation]{Proposition}

\newcommand{\ZZ}{\mathbb{Z}}
\newcommand{\CC}{\mathbb{C}}
\newcommand{\QQ}{\mathbb{Q}}
\newcommand{\RR}{\mathbb{R}}
\newcommand{\Mb}{\bar{M}}
\newcommand{\Sb}{\bar{S}}
\newcommand{\fl}{f_{\ast}}
\newcommand{\fb}{\bar{f}}
\newcommand{\fbl}{\fb_{\ast}}
\newcommand{\shH}{\mathcal{H}}
\newcommand{\shHZ}{\shH_{\ZZ}}
\newcommand{\shHQ}{\shH_{\QQ}}
\newcommand{\shHR}{\shH_{\RR}}
\newcommand{\shHC}{\shH_{\CC}}

\newcommand{\into}{\hookrightarrow}
\newcommand{\jl}{j_{\ast}}
\newcommand{\tensor}{\otimes}
\DeclareMathOperator{\Hom}{Hom}

\newcommand{\BZ}{B_{\ZZ}}
\newcommand{\Bd}{B^{\ast}}
\newcommand{\shO}{\mathscr{O}}
\newcommand{\OS}{\shO_S}
\newcommand{\varphisig}{\varphi_{\sigma}}

\newcommand{\SLZ}{SL(2, \ZZ)}
\newcommand{\eps}{\varepsilon}
\newcommand{\epsb}{\bar{\eps}}

\title{The Cremmer-Scherk Mechanism\\
in F-theory Compactifications on K3 Manifolds}

\author[a]{Michael R. Douglas,}
\author[a]{Daniel S. Park}
\author[b,c]{and Christian Schnell}

\affiliation[a]{Simons Center for Geometry and Physics\\
Stony Brook University\\
Stony Brook, NY 11794-3636, USA}
\affiliation[b]{Department of Mathematics\\
Stony Brook University\\
Stony Brook, NY 11794-3651, USA}
\affiliation[c]{Mathematisches Institut\\
Universit\"at Bonn\\
Endenicher Allee 60, 53115 Bonn, Germany}

\emailAdd{mdouglas {\rm at} scgp.stonybrook.edu}
\emailAdd{dpark {\rm at} scgp.stonybrook.edu}
\emailAdd{cschnell {\rm at} math.sunysb.edu}

\abstract{It is well understood --- through string dualities ---
that there are 20 massless vector fields in the spectrum of
eight-dimensional F-theory compactifications on smooth
elliptically fibered K3 surfaces at a generic point in the K3
moduli space. Such F-theory vacua,
which do not have any enhanced gauge
symmetries, can be thought of as supersymmetric
type IIB compactifications on $\field{P}^1$
with 24 $(p,q)$ seven-branes. 
Naively, one might expect there to be 24 massless vector
fields in the eight-dimensional effective theory
coming from world-volume gauge fields of the 24
branes.
In this paper, we show how the vector field
spectrum of the eight-dimensional effective theory 
can be obtained from the point of view of type IIB
supergravity coupled to the world-volume theory
of the seven-branes.
In particular, we first show that
the two-forms of the type IIB theory
absorb the seven-brane world-volume
gauge fields via the
Cremmer-Scherk mechanism.
We then proceed to show that the massless
vector fields of the eight-dimensional theory
come from KK-reducing the $SL(2,\field{Z})$
doublet two-forms of type IIB theory
along $SL(2,\field{Z})$ doublet one-forms
on the $\field{P}^1$. We also discuss the relation
between these vector fields and the ``eaten"
world-volume vector fields of the seven-branes.}

\begin{document}
\maketitle
\flushbottom

\section{Introduction and Summary}

Ever since its discovery, F-theory \cite{Vafa:1996xn,
Morrison:1996na,Morrison:1996pp}
has played a prominent role in understanding the
landscape of string vacua. F-theory provides
a very rich, if not the richest, range of string vacua
in various dimensions.\footnote{See, for example,
\cite{Taylor:2011wt}.} This ``versatility" comes from
the fact that F-theory provides a framework to work
with strongly coupled string --- or to be exact,
type IIB --- backgrounds.
An F-theory vacuum can be thought of as
a compactification of a twelve-dimensional theory
on an elliptically fibered manifold $\bar M$
over some base $\bar S$.
What this background is actually
describing is a type IIB compactification on the manifold
$\bar S$ with a varying axio-dilaton profile
--- the value of the
axio-dilaton is encoded in the complex structure
of the elliptic fiber.

The strongly coupled nature of F-theory, however,
makes it difficult to study global F-theory backgrounds 
directly from the point of view of type IIB string theory.
There are many different approaches to understand these
vacua. One approach is to use string dualities with
M-theory or heterotic string theory \cite{Vafa:1996xn,
Morrison:1996na,Morrison:1996pp}.
Another is to study weakly coupled ``orientifold"
limits \cite{Sen:1996vd,Aluffi:2009tm,Clingher:2012rg}
of F-theory vacua.
Yet another is to study local backgrounds to gain insight
into global backgrounds
\cite{Donagi:2008ca,Beasley:2008dc,Beasley:2008kw,
Donagi:2008kj}.
By now there are many aspects of F-theory that are
well understood based on these approaches.
Some features, however, remain unclarified from the
point of view of type IIB string theory.

A subject that begs for better understanding is abelian gauge
symmetry. For example, the abelian gauge symmetry of
eight-dimensional or six-dimensional F-theory backgrounds
can be deduced using F-theory/M-theory duality. Its interpretation
in the original type IIB framework, however, has not been
explored extensively. Let us elaborate the issue with K3
compactifications of F-theory, which is the subject of this paper.

The simplest F-theory backgrounds are eight-dimensional
--- they come from compactifying the theory on
an elliptically fibered K3 manifold with a section.
When the K3 manifold only has $I_1$
singularities, these backgrounds describe type IIB
compactifications on $\field{P}^1$ with 24
$(p,q)$ seven-branes.
K3 compactifications of F-theory were thoroughly investigated
from --- and arguably even before \cite{Greene:1989ya}
--- the birth of F-theory \cite{Vafa:1996xn,Sen:1996vd,
Douglas:1996du,Banks:1996nj,LopesCardoso:1996hq,
Lerche:1998nx,DeWolfe:1998pr,Lerche:1999de}
and are very well understood
based on the aforementioned methods. In particular,
these compactifications are dual to $T^2$
compactifications of heterotic string theory, which are
perturbative string vacua.
The eight-dimensional F-theory compactification has
20 vector fields in its massless spectrum
at a generic point in the F-theory moduli space,
18 of which belong to the vector multiplets.
It is understood that these vector fields are
related to the world-volume vector fields of the
24 seven-branes present in the background.
The precise relation between the massless
vector fields of the 8D theory and the
vector fields living on the world-volume of
the seven-branes,
however, has not been explored further.
For example, while qualitative explanations
on the discrepancy between the number of
the branes and the number of vector fields in the
eight-dimensional theory have been given
\cite{Vafa:1996xn,Douglas:1996du},
these arguments have not been made very sharp.
In this paper, we expand on the idea of \cite{Douglas:1996du}
on how the vector field spectrum of F-theory compactified on
K3 can be obtained. In particular, we take the point of view that
these backgrounds are type IIB supergravity compactifications
on a $\field{P}^1$ with seven-branes in it.\footnote{This approach
to F-theory backgrounds has been utilized for different
purposes before, for example, in the works
\cite{Bergshoeff:2002mb,Bergshoeff:2006jj}.}\footnote{A similar
approach to obtaining matter spectra of
F-theory compactifications on Calabi-Yau
manifolds has been taken in
\cite{Grassi:2013kha,Grassi:2014sda}.}
We focus on the interaction of the
bulk two-form fields of the type IIB theory and the world-volume
gauge fields, ultimately showing that the gauge degrees of
freedom are eaten by the tensor fields
through the Cremmer-Scherk (CS)
mechanism\footnote{Incidentally, the initials of Cremmer-Scherk
coincide with those of Chern-Simons. Throughout this paper,
the acronym CS is exclusively used to refer to the former
combination.} \cite{Cremmer:1973mg}.

The Cremmer-Scherk mechanism is a generalized
version of the St\"uckelberg mechanism to the tensor/vector field
pair. Let us first review the St\"uckelberg mechanism before we
describe its generalized version.
An abelian vector field $A_\mu$ can become massive by
coupling to a scalar St\"uckelberg field $\phi$ by
\be
{1 \ov 2} (\p_\mu \phi - A_\mu)^2 \,.
\ee 
The gauge symmetry of the theory is given by
\be
A_\mu \ra A_\mu + \p_\mu \Lambda,\quad
\phi \ra \phi + \Lambda \,,
\ee
and hence the St\"uckelberg field can be gauged away.
In the end, the degrees of freedom of the scalar field
are eaten by the gauge field --- one is left with one massive
vector field in the theory.

One can readily generalize this mechanism for tensor-vector
interactions. That is, given a two-form field $B_{\mu \nu}$ and
vector field $A_\mu$, the two-form field becomes massive by
the covariant coupling
\be
{1 \ov 2} (\p_\mu A_\nu -\p_\nu A_\mu - B_{\mu \nu})^2 \,,
\label{t-vint}
\ee
given the gauge symmetry
\be
B_{\mu \nu} \ra B_{\mu \nu} + \p_\mu V_\nu -\p_\nu V_\mu ,\quad
A_\mu \ra A_\mu +V_\mu \,.
\label{t-vgauge}
\ee
Now the vector is ``eaten" by the tensor field ---
this is the Cremmer-Scherk mechanism.
The tensor-vector interaction \eq{t-vint} and the gauge symmetry
\eq{t-vgauge} is a familiar one --- it is precisely the way gauge fields
living on branes interact with bulk tensor fields.
The two-forms of the type IIB theory,
which form a doublet under the global
$SL(2,\field{Z})$ action of the theory,
couple to world-volume gauge fields in this way.
In this paper, we show that for
F-theory K3 compactifications ---
type IIB compactifications
on $\field{P}^1$ with 24 seven-branes
--- all the 24 world-volume gauge fields are ``eaten"
by the two-form fields.
More precisely, we find 24 linearly independent
$SL(2,\field{Z})$ doublet gauge transformations
\be
B^I_{MN} \ra B^I_{MN} +
d(\phi^I_a (z,\bar{z})  \Lambda^a_{\nu}), \quad
A^i_\mu \ra A^i_{\mu} -
(p_i \phi^1_a (z_i,\bar{z_i}) 
+q_i \phi^2_a (z_i,\bar{z_i}) )  \Lambda^a_{\mu}
\ee
where $a=1,\cdots,24$.
Here, $M, N$/$\mu, \nu$ are ten/eight-dimensional indices,
respectively, while $z$ and $\bar{z}$ denote the coordinates
on the internal $\field{P}^1$ manifold. The $I$ is an
$SL(2,\field{Z})$ index.
We have used $i=1, \cdots,24$ to index the branes,
while using $(z_i,\bar{z}_i)$ and $(p_i,q_i)$ to
denote their positions and brane
charges, respectively. $A_i$ is the gauge field living on
the $i$-th brane.
Hence due to the CS gauge symmetry of the system,
we can work in a ``unitary gauge" where the vector degrees
of freedom are pulled from the branes into the ``bulk."

Although the world-volume gauge fields are
eaten by the tensor fields through the CS mechanism,
it turns out that there are still massless vector fields ---
in fact, 20 of them ---
in the eight-dimensional effective theory.
These vector fields come from KK-reducing the
$SL(2,\field{Z})$ doublet two-form along
$SL(2,\field{Z})$ doublet one-form zero
modes on the compact $\field{P}^1$:
\be
B^I_{m \mu} =\sum_{k=1}^{20} \xi^{k,I}_{m} a^{k}_\mu\,.
\ee
Here, $m$ is the two-dimensional index along
the compact $\field{P}^1$ direction, while we have
used $k$ to enumerate the zero modes.
These zero modes are harmonic along the compact
directions, i.e.,
\be
d \xi^{k,I} =0 \,, \quad
d * \MM_{IJ }\xi^{k,J} =0 \,,
\label{harmonic}
\ee
while they must exhibit certain monodromies
around seven-brane loci. Here, $*$ denotes the
Hodge dual with respect to the metric of the base
manifold, while $\MM_{IJ}$
is a $SL(2,\field{Z})$ covariant Hermitian metric
which depends on the axio-dilaton.
We count the number of these zero-modes
by relating them to elements of the cohomology group
of a certain sheaf living on the base $\bar{S}$ of the
elliptic fibration.

To introduce this sheaf,
let us review the F-theory backgrounds
at hand in more detail.
As before, let us denote the 24 seven-branes as $B_i$ with
$i=1, \cdots, 24$. From the point of view of the K3 geometry,
these branes sit at the loci of the base $\field{P}^1$ where
the elliptic fiber degenerates.
Picking an ``$A$-cycle" $\alpha$ and a ``$B$-cycle"
$\beta$ along the fiber, we can determine the type of
brane sitting at $B_i$.\footnote{Such a choice can always
be made in a dense open patch of the base of the fibration.}
$B_i$ is a $(p_i,q_i)$ seven-brane when the
cycle  $p_i \a+ q_i \b$ degenerates at the brane locus.
Now the $A$ and $B$-cycle exhibit monodromies
around the brane locus --- these are precisely
the monodromies that $SL(2,\field{Z})$
covariant fields must exhibit around the branes
in order for the field values to be well-defined.

We see that one way to view the K3 manifold is
to see it as a family of elliptic curves parametrized by the
base manifold $\bar S$. From this point of view, the
harmonic one-forms \eq{harmonic} represent
elements of the first cohomology group of ``the sheaf of
local invariant one-cycles"\footnote{Given the elliptic
fibration $\bar f: \bar M \ra \bar S$, this sheaf is
obtained by pushing forward a certain sheaf living
in a dense open subset $S$ of $\bar S$ with respect
to the inclusion map $j: S \hookrightarrow \bar S$.
$S$ is obtained by excising the points on $\bar S$
where the fiber degenerates. Then one can consider
the elliptic fibration $f: M \ra S$ over $S$.
The sheaf living on $S$ is denoted by
$\shHQ=R^1 f_* \QQ$ for the locally
constant sheaf $\QQ$ on $M$. The cohomology group
of interest is $H^1 (\bar S, j_* \shHQ)$ \cite{Zucker}.
Explanation of the notation we use
can be found in standard texts on Hodge theory
such as \cite{Voisin}.}
living on $S$.
Hence, the dimension of this cohomology group,
$h^1 (\bar S, j_* \shHQ)$,
can be identified with the number of linearly
independent doublet harmonic one-forms.
Cohomology groups of such sheaves have
been examined systematically
in the mathematics literature \cite{Zucker},
and have been shown to have Hodge structures
compatible with that of
the elliptically fibered manifold itself.
We use the results of \cite{Zucker}
to show that $h^1 (\bar S, j_* \shHQ) =20$
(proposition \ref{prop:20}).

In this paper, we further relate the $SL(2,\field{Z})$ doublet
harmonic one-forms with the cohomology of the K3 manifold
in the following way.
We show that the doublet one-forms can be constructed
by integrating certain closed
two-forms of the underlying K3 manifold
along the $A$ and $B$-cycles of the fiber.
Let us be more precise.
There exists a 20-dimensional subspace of the
second cohomology of the
elliptically fibered K3 manifold --- which we denote
$H^2(\bar{M})_\perp$ ---
that is transverse to the fiber and the base.
For each element of $H^2 (\bar{M})_\perp$,
we show that there exists a certain
two-form $\Xi^k$ in the class
whose projection to the zero section
and to every fiber vanishes, i.e.,
\be
\Xi^k |_\text{Fiber} = \Xi^k |_\text{Base} = 0 \,.
\ee
In fact, we can choose $\Xi^k$ to be harmonic with respect
to the ``semi-flat metric" \cite{GrossWilson}
of elliptically fibered K3 manifolds
constructed in \cite{Greene:1989ya}.
In this case, a doublet of one-forms on the base manifold
\be
\begin{pmatrix}
\xi^{k,1} \\ \xi^{k,2}
\end{pmatrix}
=
\begin{pmatrix}
\int_\alpha \Xi^k \\
\int_\beta \Xi^k
\end{pmatrix}
\ee
can be defined. Note that these doublets automatically
exhibit the required monodromies around each brane
locus due to the behavior of the cycles around these
points. Also, these one-forms can be shown to be
harmonic as defined in \eq{harmonic}.
A more mathematical formulation, as well as a
proof of these facts are presented
in appendix \ref{ap:cohomsf}.

The massless vector field excitations are
equivalent to a collective excitation of
seven-brane vector fields and bulk fields
by CS gauge transformations.
A particularly useful gauge is one in which the
tensor field components are turned on along
directions transverse to the compact space.
In such a gauge, the tensor field excitations
decouple from the string junctions
\cite{DeWolfe:1998pr,Gaberdiel:1997ud,
Gaberdiel:1998mv,DeWolfe:1998zf,DeWolfe:1998eu, 
Fukae:1999zs, Huang:2013yta}
--- which are webs of $(p,q)$ strings ending
on the various seven-branes --- stretching
between the seven-branes, as the junctions
lie along the compact $\field{P}^1$.
Therefore, in this gauge, one can
identify the linear combinations
of the seven-brane vector fields
that reproduce the charges of the string
junctions under a particular vector field $a^k$.

In this sense, there is a correspondence between
the massless vector fields constructed by KK-reduction
and the world-volume vector fields living on the seven-branes.
To be more precise, it can be shown that
turning on an eight-dimensional vector field $a_k$, i.e.,
turning on the ten-dimensional tensor field
\be
B^I = \xi^{k,I} \wedge a^{k}\,,
\label{actB}
\ee
is gauge equivalent to turning on some linear
combination of seven-brane vector fields
\be
A^i = \Phi^i_k a^k \equiv
(p_i \varphi^{k,1} (z_i,\bar{z_i}) 
+q_i \varphi^{k,2} (z_i,\bar{z_i}) )  a^k \,,
\label{ID}
\ee
along with a tensor field
transverse to the compact directions:
\be
B^I =
- \varphi^{k,I} d a^k \,.
\ee
Here, $\varphi^{k,I}$ is a doublet scalar
living on the $\field{P}^1$ that satisfies
\be
d \varphi^{k,I} = \xi^{k,l} \,.
\label{gt}
\ee
Hence the tensor field background \eq{actB}
is equivalent to turning on the background
gauge fields \eq{ID} from the point of view
of the string junctions.
With further ``CS gauge fixing," we can
in fact show that there is a invertible linear map
between the vector fields $a_k$ and a
moduli-independent 20-dimensional
linear subspace $L$ of the
seven-brane vector fields.

This paper is organized as follows. In section \ref{s:review},
we review basic facts about K3 compactifications of F-theory
and show how they can be described
from the type IIB point of view.
In section \ref{s:stuckelberg}, we show that all the
seven-brane world-volume vector fields
are eaten by the type IIB tensor fields through
the Cremmer-Scherk mechanism, and identify the
responsible gauge transformations.
In section \ref{s:doublet}, we find the 20 $SL(2,\field{Z})$
doublet harmonic one-forms of the type IIB geometry.
We relate these one-forms to the elements of
the cohomology group $H^1 (\bar S, j_* \shHQ)$,
as well as $H^2 (\bar M)$.
We show that the type IIB doublet two-forms can be
reduced along these one-forms to yield 20 massless
vector fields in the 8D effective theory.
We proceed to establish the aforementioned
correspondence between these
harmonic one-forms and world-volume vector fields.
Further discussions and future directions
are presented in section \ref{s:future}.
In particular, we discuss the possibility of developing our
approach further towards
understanding more general
F-theory backgrounds.
We elaborate on some technical details that
we have omitted in the main text in the appendix.

\section{Review of F-theory Compactifications on Smooth K3 Surfaces}
\label{s:review}

In this section, we review eight-dimensional
backgrounds coming from compactifying
F-theory on a smooth generic elliptically fibered
K3 manifold with a section.\footnote{By smooth and generic,
we mean that the elliptic fibration has 24 $I_1$
singularities. We note that smooth
K3 manifolds can have type $II$
singularities at special points in the moduli space.}
We we take the point of view that these backgrounds
are supersymmetric solutions of type IIB theory,
that is, as a type IIB compactification on
$\field{P}^1$ with 24 $(p,q)$ seven-branes.
We first review its massless matter content
using F-theory/heterotic duality, focusing on the gauge fields,
and proceed to describe the supergravity solution
in more detail.
The content of this section is a reorganization of
facts presented in \cite{Vafa:1996xn,Taylor:2011wt,
Greene:1989ya,LopesCardoso:1996hq,DeWolfe:1998pr,
Lerche:1999de,Gaberdiel:1997ud,Gaberdiel:1998mv,
DeWolfe:1998zf,DeWolfe:1998eu,Fukae:1999zs,
Huang:2013yta},
among other places.
A great review of K3 geometry in the context of string theory
is given in \cite{Aspinwall:1996mn}.

Eight-dimensional F-theory backgrounds with minimal
supersymmetry come from compactifying F-theory on an
elliptically fibered K3 manifold $\bar f : \bar{M} \ra \bar{S}$ with a section.
We denote the K3 manifold by $\bar{M}$ and the base manifold by $\bar{S}$
throughout this paper.
The base $\bar{S}$ of the fibration is a $\field{P}^1$, and 
the manifold is parametrized by the Weierstrass equation
\be
Y^2 = X^3 + F_{8} X W^4 + G_{12} W^6 \,.
\label{K3eq}
\ee
Here $F_8$ and $G_{12}$ are sections of $8H$ and $12H$,
where $H$ is the hyperplane line bundle of the base $\field{P}^1$
manifold.
In this paper, we assume that the complex structure of $\bar{M}$ is
at a generic point in the moduli space.
We therefore assume
generic values for the coefficients of $F_8$ and $G_{12}$.
When this is the case, the manifold is smooth
and the elliptic fibration has 24 $I_1$ singularities at
the loci
\be
\Delta = 4F_8^3+27G_{12}^2 =0 \,.
\ee
$\Delta$ is the discriminant of the elliptic curve; the locus
$\Delta=0$ is called the discriminant locus.
Throughout this paper, we often choose
work in a local patch of the ambient toric manifold, in which case
the equation \eq{K3eq} can be written as
\be
y^2 = x^3 + f_8 (z) x+ g_{12}(z) \,,
\label{K3loc}
\ee
where $z$ is the local coordinate on the base manifold.
$f_8$ and $g_{12}$ are polynomials in $z$ with
degree $\leq \! 8$ and $\leq \!12$, respectively.
There are thirty-seven moduli in the eight-dimensional theory ---
18 complex moduli parametrizing the complex structure 
of the elliptic fibration \eq{K3eq} and one real modulus that
parametrizes the size of the base.

These eight-dimensional theories are dual to $E_8 \times E_8$
heterotic string compactifications on a two-torus.
The complex structure moduli of the elliptically fibered
K3 manifold map to the complex and K\"ahler
moduli of the torus and the Wilson lines
along the two $T^2$ directions.
The modulus that parametrizes the size of the
base of the K3 manifold maps to the value of the
dilaton of the heterotic theory.
At a generic point in the complex structure moduli space,
the dual heterotic background has generic Wilson
lines turned on.
The massless spectrum of the heterotic background
can be easily obtained by standard methods.
In particular, the theory at such a point has 20
gauge fields in the massless spectrum.
Sixteen of these gauge fields come
from the Cartan subgroup of the $E_8 \times E_8$ gauge
group, while four --- two of which are graviphotons ---
come from reducing the ten-dimensional
graviton and tensor along the two ``legs" of the torus.

The elliptic fibration \eq{K3eq} describes a
supersymmetric background of type IIB string theory.
Before we see how, let us first
describe the low-energy effective theory of
type IIB in more detail.
The massless bosonic degrees of freedom are given
by the graviton, a complex scalar, two two-forms,
and one self-dual four form.
Type IIB string theory is covariant under
a global $SL(2,\field{Z})$ group.
Following the conventions of \cite{Polchinski:1998rr},
the bosonic part of the type IIB action can be written in
an $SL(2,\field{Z})$ covariant way in
Einstein frame:
\be
S_{IIB} = {1 \ov 2 \kappa_{10}^2} \int d^{10} x
\sqrt{- g}(R - {\p_\mu \tau \p^\mu \bar{\tau} \ov 2 \tau_2^2}
-\M_{IJ} F_3^I \cdot F_3^J -{1 \ov 4} |\tilde {F}_5^2|)
-{\epsilon_{IJ} \ov 8 \kappa_{10}^2}  \int C_4 \wedge F_3^I \wedge F_3^J \,.
\label{IIB}
\ee
Here, $\tau = \tau_1 + i \tau_2$ is the axio-dilaton,
while $F_3^I$ is the two-from field strength doublet:
\be
\begin{pmatrix}
F_3^1 \\
F_3^2
\end{pmatrix}
=
\begin{pmatrix}
dB \\
dC
\end{pmatrix} \,.
\ee
The $SL(2,\field{Z})$ group acts on these fields as
\begin{align}
\tau &\ra {a \tau + b \ov c \tau +d} \\
F_3^I &\ra \Lambda^I_J F_3^J, \quad \Lambda^I_J =
\begin{pmatrix}
d&c \\ b&a
\end{pmatrix}
\label{sl2z}
\end{align}
where $a,b,c$ and $d$ are integers satisfying $ad-bc =1$.
We note that the dual six-form fields of $B$ and $C$
transform in the same way as the two-form fields under the
$SL(2,\field{Z})$ action.
The matrix $\M_{IJ}$ is given by
\be
\M_{IJ} = {1 \ov \tau_2}
\begin{pmatrix}
|\tau|^2&-\tau_1 \\
-\tau_1&1
\end{pmatrix} \,.
\label{mij}
\ee
$\tilde{F}_5$ is the $SL(2,\field{Z})$ neutral four-form field
strength.
Using the transformation rules, it can be checked
that the action \eq{IIB} is invariant under
$SL(2,\field{Z})$ transformations.

From the point of view of the type IIB theory,
the elliptic fibration \eq{K3eq} parametrizes a
supersymmetric solution to the equations of motion.
To be more precise, it describes a compactification of type IIB
theory on a $\field{P}^1$ with a varying axio-dilaton.
Taking the base of the elliptic fibration \eq{K3eq}
to be the compact $\field{P}^1$, the axio-dilaton
value at a point $z$ in the base is related to the complex
structure $\tau(z)$ of the fiber at the given point by
\cite{Greene:1989ya}
\be
j(\tau) = 1728 \times {4 F_8^3 \ov 4F_8^3 + 27 G_{12}^2}
\,.
\label{j}
\ee
$j$ is Klein's $j$-invariant.
The metric on the $\field{P}^1$ can also be computed
from the elliptic fibration \cite{Greene:1989ya}
\be
ds^2 =  {\tau_2 |\eta(\tau)|^4 \ov |\Delta|^{1/6}} dz d \bz \,,
\label{metric}
\ee
where $z$ is the complex coordinate along the
base $\field{P}^1$.\footnote{Throughout this paper,
we use $z$ and $\bar{z}$ to denote
the internal coordinates of the type IIB compactification.
The coordinates along the non-compact direction is denoted
by $x^\mu$.}
The 24 loci where the fiber degenerates
can be thought of as seven-brane loci.
Let us denote these branes as $B_1, \cdots B_{24}$.

Now the $j$-invariant \eq{j} is not
a one-to-one function from the upper-half complex plane
to the complex plane. In order to describe the
F-theory vacuum from the point of view of type IIB,
one must also choose two one-cycles ---
the $A$-cycle $\a$ and the $B$-cycle $\b$ ---
of the elliptic fiber that satisfy
\be
\a \cap \a = \b \cap \b=0, \quad \a \cap \b
=- \b \cap \a =1
\label{ab}
\ee
and to compute
\be
\tau = {\int_\b \lambda \ov \int_\a \lambda} \,,
\ee
where $\lambda$ is the unique holomorphic one-form
on the elliptic fiber.\footnote{An $A$-cycle and $B$-cycle
of an elliptic curve are, in fact,
defined to be a pair of one-cycles
that satisfy the very relations \eq{ab}.}
The choice of different pairs of cycles
that satisfy the conditions \eq{ab}
result in different type IIB backgrounds related by
$SL(2,\field{Z})$ transformations. The group of
global $SL(2,\field{Z})$ transformations is nothing but the
group of maps between different choices of cycles.

\begin{figure}[!t]
\centering\includegraphics[width=5cm]{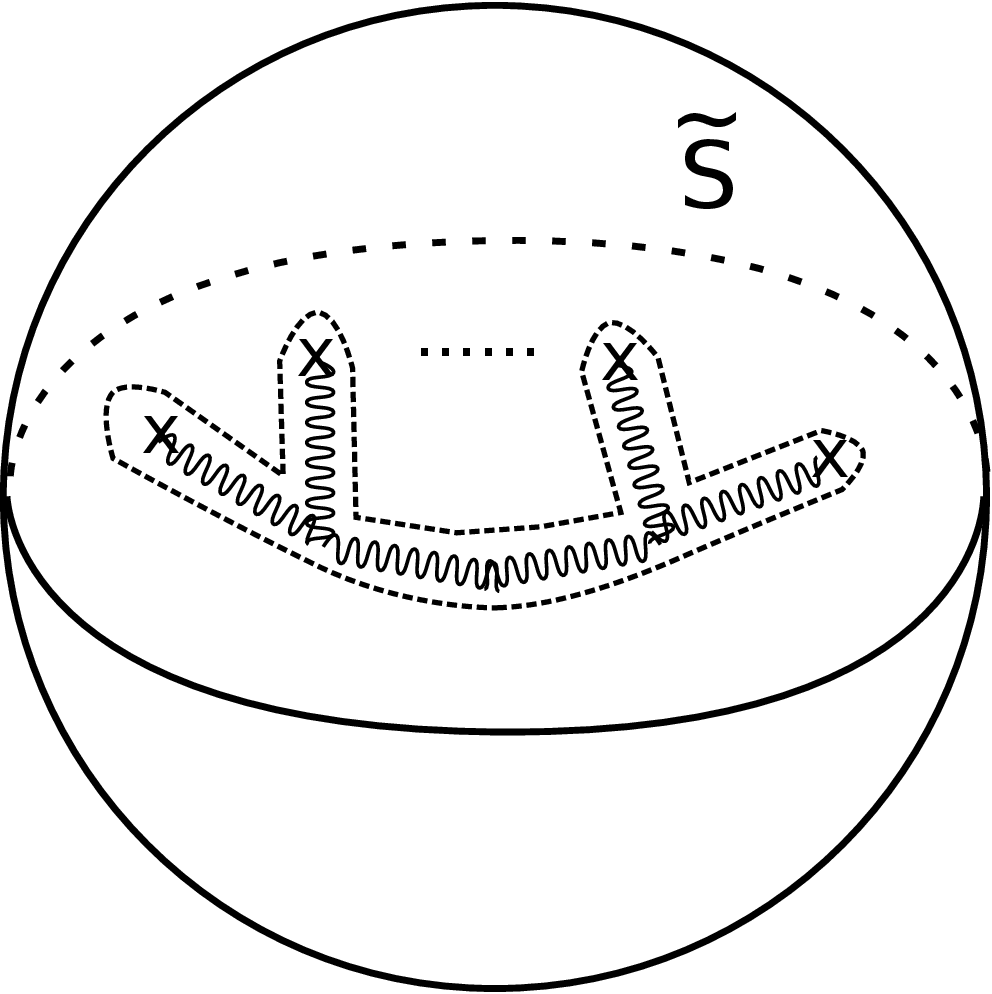}
\caption{\small A depiction of the $\field{P}^1$ base of an
elliptically fibered K3 manifold. The marked points denote
the seven-brane loci at which the fiber degenerates.
The cycles of the elliptic fiber undergo monodromies around
these points, and hence a global definition of an $A$-cycle
and a $B$-cycle of a fiber does not exist. We can, however,
define $A$ and $B$ cycles when we exclude branch cuts
--- depicted as wavy lines --- emanating from the seven-brane loci,
{\it i.e.,} when outside the region encircled by the dotted lines.
The $A$ and $B$ cycles are well defined in this dense open
subset, which we denote $\tilde S$.}
\label{f:P1}
\end{figure}

The cycles of the elliptic fiber undergo
monodromies as they go around the
seven-brane loci --- the value of the axio-dilaton
transforms under the corresponding monodromies accordingly.
Therefore the axio-dilaton profile of a non-trivial F-theory
background cannot be defined globally on the base
manifold --- in fact, there are branch cuts emanating from the
seven-brane loci.
In the case the elliptically fibered
manifold is a K3 manifold, the overall monodromy is
trivial.
Therefore we can ``join" the 24 branch cuts
emanating from each brane. We can then define the
$A$-cycle $\a$ and $B$-cycle $\b$
of the elliptic fibration unambiguously
in the dense open subset $\tilde{S}$
of the base $\field{P}^1$ manifold
obtained by excluding these branch cuts.
We note that the monodromy around
each brane --- and hence the type of each brane ---
depends on how one chooses these branch cuts
\cite{Gaberdiel:1998mv}.\footnote{In fact,
one can only make sense of the monodromies
as being an element of $SL(2,\field{Z})$ when $A$ and
$B$-cycles can be defined. Therefore a set of branes
can have many different representations as $(p,q)$-branes
depending on how one decides to ``join" the cuts emanating
from them. It is useful to note that two different $(p,q)$ brane
configurations obtained by choosing different ways of joining cuts
are not in general related to each other by a global $SL(2,\field{Z})$
transformation.
Such equivalences between different $(p,q)$-brane
configurations have been extensively studied from the point
of view of string junctions.}
Unless an F-theory background has an orientifold limit,
we must always pick such a patch
to describe the backgrounds in the type IIB framework.
In this sense, a useful way to view these F-theory backgrounds
is to interpret them as type IIB compactifications on a dense open
subset of $\field{P}^1$ rather than the full $\field{P}^1$.
We have depicted the situation in figure \ref{f:P1}.

We note that regardless of the way one chooses the cuts,
the physics of the eight-dimensional effective theory
stays the same. Now the type IIB description of a given
compactification can alter under drastic shifting of
these cuts.
For example, when one moves a cut through a
brane locus so that their relative positions change,
the $(p,q)$ charge of the brane typically jumps.
Under local variations of the cut, however,
where no such ``singular" shifts are made, the description
of the background in terms of type IIB theory should
remain invariant. This point turns out to be important
in determining the monodromies of various fields
of the type IIB theory.

It is worth noting that a choice of cuts
defines an $SL(2,\field{Z})$ bundle on
another dense open set $S$ of the base, where
\be
S=\field{P}^1\setminus \{B_1, \cdots,B_{24}\} \,.
\ee
The way to construct the bundle is the following.
Let us choose to join branch cuts so that the tree
of branch cuts only has trivalent vertices.
Each edge of the tree has an assigned element
of $SL(2,\field{Z})$ that corresponds to the
``monodromy" that would occur from crossing that cut
in a designated direction.
For every vertex of the tree of branch cuts,
the clockwise ``monodromies" $m_i$ $(i=1,2,3)$
of the three cuts joining at the vertex must
satisfy the condition
\be
m_1m_2m_3 = \rm{id} \,.
\ee
This data corresponds to the transition
functions of a principal $SL(2,\field{Z})$ bundle
of the 24-punctured sphere $S$. We study
the sheaf associated to this bundle in detail later on.

\begin{figure}[!t]
\centering\includegraphics[width=9cm]{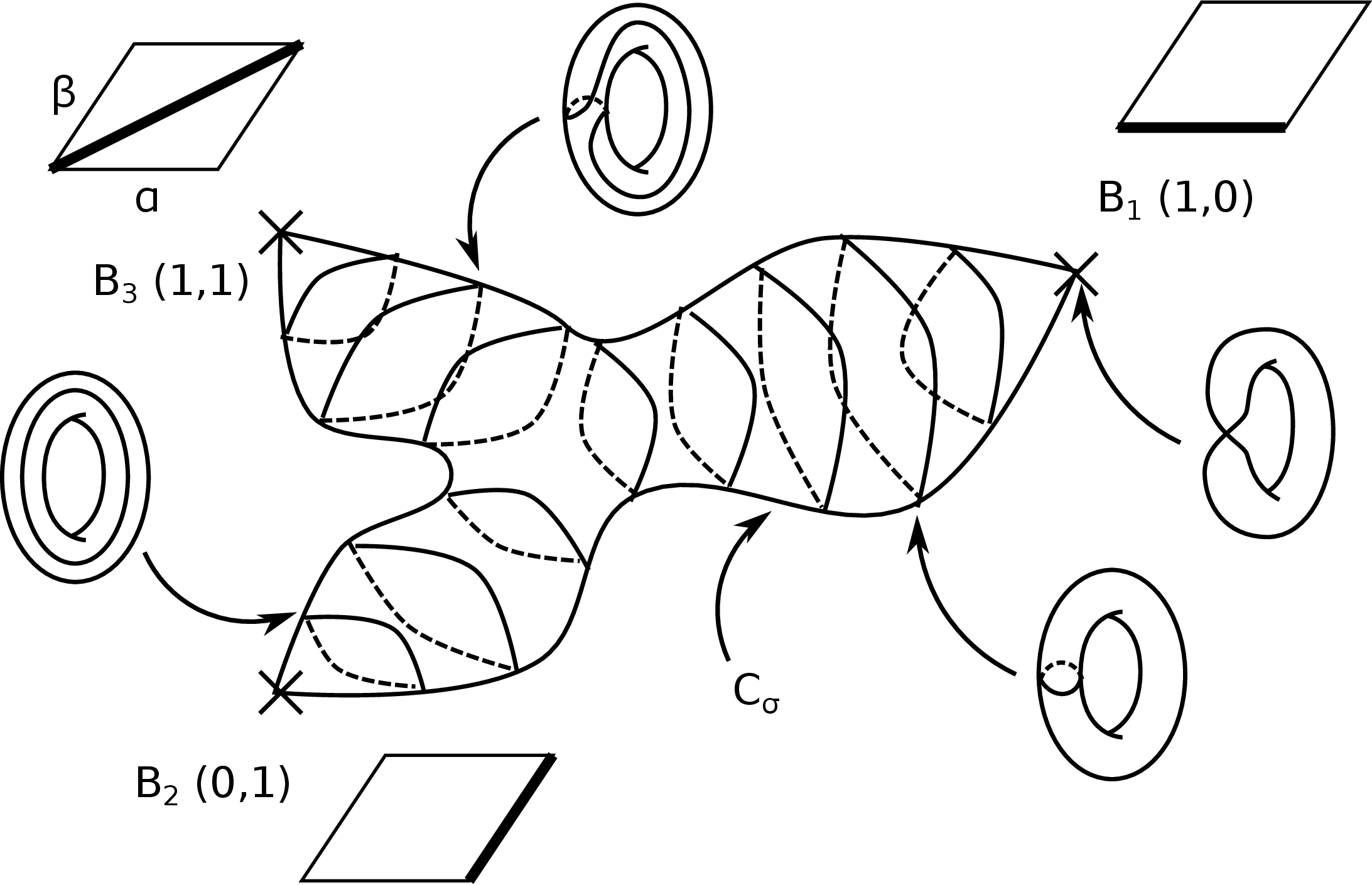}
\caption{\small An example of a two-cycle in an elliptically
fibered manifold ending at seven-brane loci, for a system
of three seven-branes $B_1$, $B_2$ and $B_3$.
The three branes are are of type
$(1,0)$, $(0,1)$ and $(1,1)$ respectively. At these branes,
the cycles $\a$, $\b$ and $\a+\b$ of the elliptic fiber
degenerate. As $(1,0)+(0,1)-(1,1)=(0,0)$,
there is a two-cycle $C_{\sigma}$
ending at the three branes with $\sigma=(1,1,-1)$.
The cycle starts off at $B_3$ where the degenerate cycle $(\alpha+\beta)$
shrinks to a point. As we trace the cycle $(\alpha+\beta)$ through the manifold
starting from $B_3$, the cycle eventually splits into cycles
$\alpha$ and $\beta$, which each shrink to a point where the branes
$B_1$ and $B_2$ are located --- the two-dimensional surface traced
out in the process defines a closed two-cycle in the elliptically
fibered manifold. The contours near each brane depict the cycles
that shrink at the corresponding brane.}
\label{f:pqdeg}
\end{figure}

As noted before,
choosing the cycles on the open set $\tilde{S}$ corresponds to
fixing a type IIB frame. Once the frame is fixed, the
types of the seven-branes at each degeneration point
can be determined. The seven-brane sitting at the point
where an irreducible cycle $p \a + q \b$ is shrinking is defined
to be a $(p,q)$ brane --- $p$ and $q$ must be mutually prime.
The monodromy around a $(p,q)$ brane is given in
the following way. The cycle $x\a+y\b$ transforms as
\be
\begin{pmatrix}
x \\
y
\end{pmatrix}
\ra
\begin{pmatrix}
1-pq & p^2 \\
-q^2 & 1+pq
\end{pmatrix}
\begin{pmatrix}
x \\
y
\end{pmatrix} \,,
\label{pqbrane}
\ee
upon rotating the elliptic fiber a full cycle
in the counter-clockwise direction around the $(p,q)$ brane.
Note that the vector $(p,q)^t$
--- representing the shrinking cycle at the seven-brane locus ---
is left invariant by this monodromy.
The $(1,0)$ brane is a D7-brane where a fundamental
string can end at, while D1-strings can end at $(0,1)$ branes.
In fact, $(p,q)$ seven-branes are defined to be seven-branes
at which $(p,q)$ string can end. Let us denote the brane charge
of each seven-brane $B_i$ as $(p_i,q_i)$.

Let us denote a 24 dimensional vector
$\sig = (\sig_1,\cdots,\sig_{24}) \in \field{Z}^{24}$ with
\be
\sum_i \sigma_i (p_{i},q_{i}) = (0,0)
\ee
a ``charge vector."\footnote{To use string junction
terminology, the $\field{Z}^{24}$ lattice
is the ``junction lattice" while our charge vectors are
``charge vectors of localized junctions."}
We note that the vector space of charge vectors is a
22 dimensional subspace of $\field{Z}^{24}$ due to
the two constraints.
For any charge vector $\sigma$,
there is an oriented two-cycle that begins at the
branes with $\sigma_i <0$ and ends
at branes with $\sigma_i>0$.
The end points of the cycle can be identified to be the singular
point of the fiber at the seven-brane locus, where a cycle of the
fiber is shrunk to a point.
As we move along the open patch of the base manifold,
we can trace the trajectory of such a one-cycle.
As we do so, the cycles split and merge, thereby tracing the locus
of the corresponding two-cycle inside the
elliptically fibered manifold.
$\sig$ uniquely determines a homology class of a
two-cycle. Let us denote this cycle $C_{\sigma}$.
Note that $C_{\sigma}$ is defined such that $|\sig_i|$
points of the cycle either end at ($\sig_i >0$) or begin
at ($\sig_i < 0$) the degenerate point of the fiber above
brane $B_i$ when $\sig_i \neq 0$. Examples of $C_\sigma$
for two different $\sig$ are given in figures \ref{f:pqdeg}
and \ref{f:mult}.

\begin{figure}[!t]
\centering\includegraphics[width=9cm]{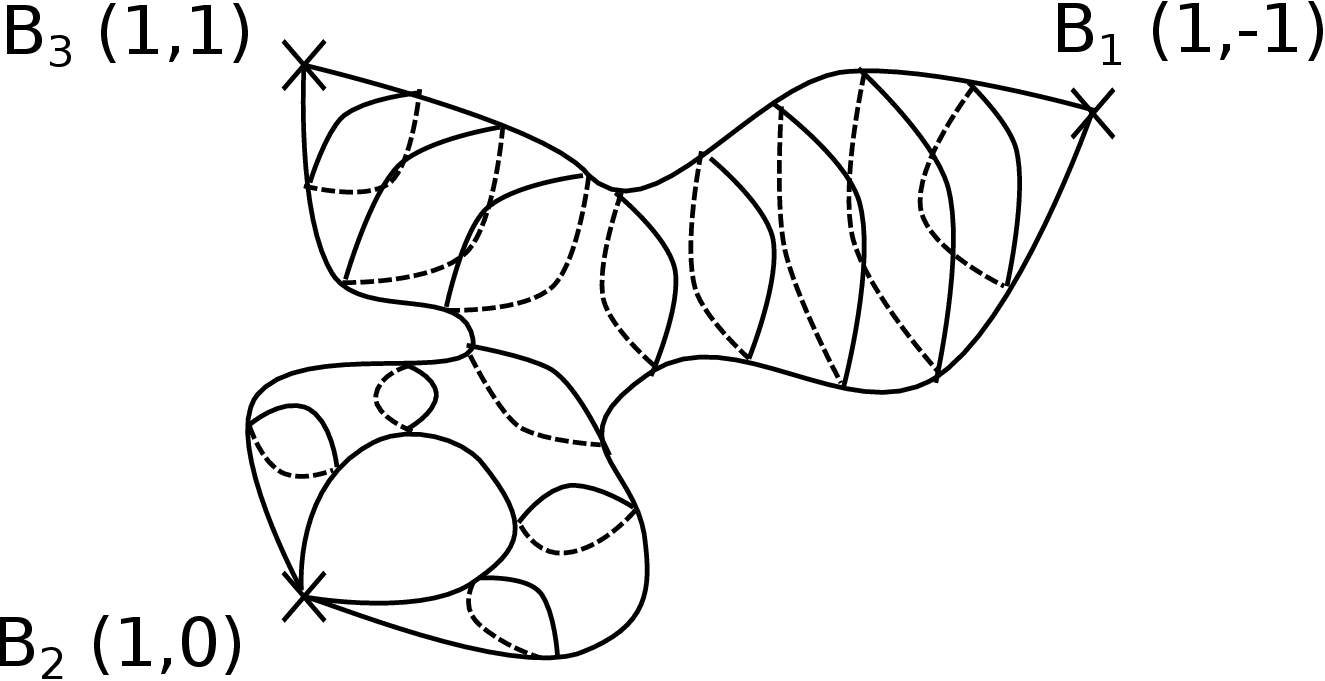}
\caption{\small When $|\sig_i|>1$,
multiple points of $C_\sig$ end or begin at
the brane locus $B_i$.
This image depicts $C_\sig$ for $\sig=(1,-2,1)$
for the system of branes $B_1$, $B_2$ and $B_3$
with branes charges $(1,-1)$, $(1,0)$
and $(1,1)$ respectively.}
\label{f:mult}
\end{figure}

Let us note that
\be
C_{\sig} + C_{\sig'} = C_{\sig''} \quad
{\rm when} \quad
\sig + \sig' = \sig''
\label{homeq}
\ee
where the former equality means that the cycles
are equal as homology classes.
This can be easily confirmed, as $C_{\sig} + C_{\sig'}$
can be smoothly deformed into $C_{\sig''}$ when
$\sig + \sig' = \sig''$. The inverse statement, however,
is not true. This is because two non-trivial charge
vectors correspond to trivial homology classes.\footnote{These
charge vectors are referred to as ``zero" or
``null vectors" in the string junction literature.}

The homology group of two-cycles $C_\sig$
--- {\it i.e.,} cycles
that interpolate between seven-brane loci ---
is in fact generated by 20 elements $C_1, \cdots, C_{20}$.
This is because the second homology group of a K3 manifold,
when viewed as a vector space, is 22 dimensional
--- one of which corresponds to the
class of the fiber and one of which corresponds to the class
of the base\footnote{By ``class of the base" we are actually
referring to the ``class of the section." Since we always assume
the existence of a section, we do not make the effort of
distinguishing the terminology.}.
The two-cycles we are interested in are generated
by the elements that are orthogonal --- with respect to the
intersection product --- to the base and fiber classes.
Let us denote this 20-dimensional space as $H_2 (\bar M)_\perp$.
The complex structure of an elliptically fibered K3 manifold
is determined by the ray of the complex vector
\be
(\int_{C_1} \Om, \cdots, \int_{C_{20}} \Om )
\label{vect}
\ee
where $\Om$ is the holomorphic two-form of the K3 manifold
that is unique (up to a factor).
Using the local coordinates \eq{K3loc}, $\Om$
can be explicitly written as
\be
\Om = {dx dz \ov y} \,.
\ee
The complex structure of a generic K3 manifold ---
one without the restriction of being elliptically fibered or
having a section --- is parametrized by a 22 dimensional
projective vector, obtained by integrating the complex two-form
over all generators of the homology group.
The K3 manifolds used for F-theory compactifications,
however, can be parametrized by the 20 dimensional
projective vector \eq{vect} due to the fact that
$\Om$ is orthogonal to the fiber and base directions, {\it i.e.,}
\be
\int_\text{Base} \Om = \int_\text{Fiber} \Om = 0 \,.
\ee
Despite that we have parameterized the complex structure of
an elliptically fibered K3 manifold by 20 projective coordinates,
its moduli space is 18-dimensional.
This follows from the fact that the vector \eq{vect} also
satisfies the constraint \cite{Aspinwall:1996mn}
\be
\int_{K3} \Om \wedge \Om =0 \,.
\ee

\begin{figure}[!t]
\centering\includegraphics[width=12cm]{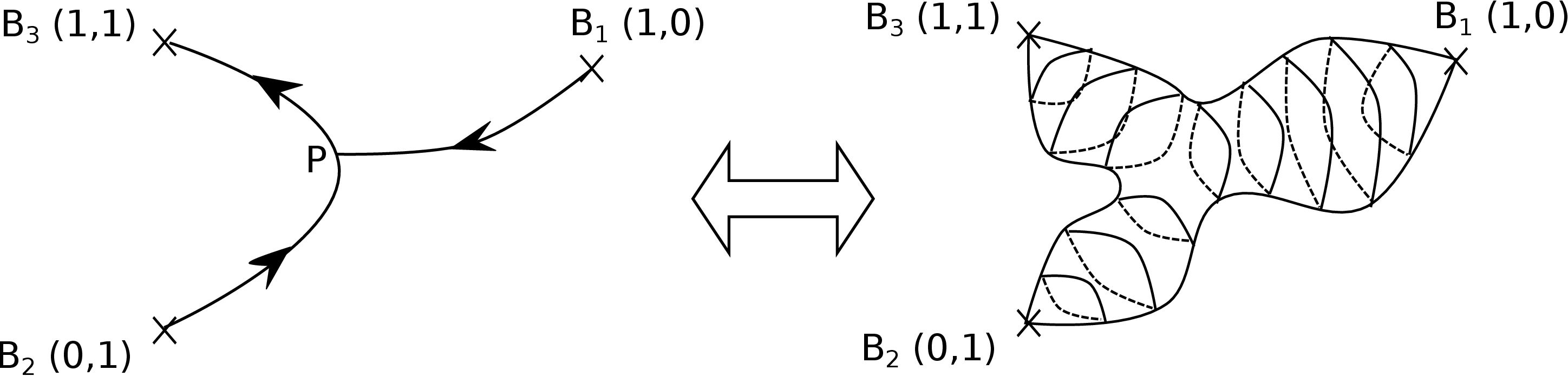}
\caption{\small Correspondence between string junctions
and two-cycles in the elliptically fibered manifold.
On the left side, we have
depicted a string junction meeting at junction point $P$ and
ending at three seven-branes,
each of type $(1,0)$, $(0,1)$ and $(1,1)$. One can ``fatten"
this junction to obtain an oriented two-cycle inside the elliptically
fibered manifold in the F-theory picture.}
\label{f:junc}
\end{figure}

The cycles $C_\sig$, as in figure \ref{f:junc}, can be
``thinned down" to a linear combination of
string junctions
--- a collection of directed $(p,q)$ string
segments connected to each other at nodes ---
stretched between the seven-branes
with net charge vector $\sig$.
Massive states of the eight-dimensional theory
can be obtained by quantizing
modes of these string junctions
and linear combinations thereof.
The correspondence between homologically
non-trivial two-cycles of the K3 manifold and
string junctions therefore implies the
correspondence between the two-cycles and the
massive particles of the 8D effective theory.

Let us end this section by commenting on the world-volume
theory living on a seven-brane. A seven-brane contains a
dynamical gauge field coming from quantizing strings with
both ends ending on that brane. The $(p,q)$ seven-brane
action can be written by first writing the D7-brane action
in Einstein frame and performing an $SL(2,\field{Z})$ transformation.
The linear terms relevant to investigating lifting of gauge fields
come from the kinetic term.
Being careful with the dilaton coupling, one can show that
the gauge kinetic term for a $(p,q)$ brane is given by
\be
-{\mu_8 \ov 4} \int d^8 x \sqrt{-g} (F_{\mu \nu}+p B_{\mu \nu}+ q C_{\mu \nu})^2
\label{pqkin}
\ee
where we have set $2\pi \alpha' = 1$
\cite{Douglas:1996du,Polchinski:1998rr}.
$\mu_8$ does not depend on the $(p,q)$ charge or the dynamical
axio-dilaton in Einstein frame.
From this, we see that the gauge coupling of the seven-branes
are independent of $(p,q)$ charge in Einstein frame ---
one could have already expected this, as type IIB theory
in Einstein frame is $SL(2,\field{Z})$ covariant.
Therefore the ten-dimensional type IIB action
corresponding to the F-theory compactification on K3
has 24 gauge fields $A^i_\mu$ living on the world volume of
seven-branes $B_i$ with the kinetic term
\be
-{\mu_8 \ov 4} \sum_{i=1}^{24}
\int d^8 x \sqrt{-g} (F^i_{\mu \nu}+p_i B_{\mu \nu}+ q_i C_{\mu \nu})^2 \,.
\label{kin24}
\ee
In the next section, we proceed to show how
these vector fields are
absorbed by the tensor fields
by the Cremmer-Scherk mechanism.

\section{The Cremmer-Scherk Mechanism} \label{s:stuckelberg}

In this section, we show that the vector fields living
on the world-volume of the seven-branes of the F-theory
background are ``eaten" by the type IIB tensor fields
through the Cremmer-Scherk mechanism.
We first review the Cremmer-Scherk gauge transformations
of F-theory compactifications on K3. We proceed to identify the
24 gauge transformations responsible for absorbing the vector
fields into the tensor degrees of freedom.

As can be seen from the previous section, type IIB backgrounds
with seven-branes are invariant under the local symmetry
\be
B^I \ra B^I +
d \Gamma^I , \quad
A^i_\mu \ra A^i_{\mu} -
\pi^i_* (p_i  \Gamma^1 (z_i,\bar{z_i}) 
+q_i \Gamma^2 (z_i,\bar{z_i}) )
\label{CSK3}
\ee
where $\Gamma^I$ is a doublet of one-forms.
As before, $I$ is the $SL(2,\field{Z})$ index and
$i$ indexes the branes. Also, $\pi^i_*$ is the push-forward
of the projection map to the $i$'th brane.
In order for \eq{CSK3} to make sense,
a constraint on the doublet one-form $\Gamma^I$
must be imposed. The value
\be
\pi^i_* (p_i  \Gamma^1 (z_i,\bar{z_i}) 
+q_i \Gamma^2 (z_i,\bar{z_i}) )
\label{gtonbloc}
\ee
must be unambiguously defined, i.e.,
it must be monodromy invariant at
each brane locus.

We stress that while the two-form field strength
must exhibit certain monodromies,
$\Gamma^I$ need not show such behavior.
Upon investigation of the Lagrangian of the theory,
one can verify that the field strengths
of the two-form fields
$F^I = dB^I$ are required to undergo
monodromies
\be
\begin{pmatrix}
dB \\ dC
\end{pmatrix}
\ra \begin{pmatrix}
1-pq &-q^2 \\
p^2 & 1+pq
\end{pmatrix}
\begin{pmatrix}
dB \\ dC
\end{pmatrix}
\label{2fmon}
\ee
upon counter-clockwise rotation around
--- or, equivalently, crossing the branch cut
in clockwise direction --- a $(p,q)$ brane.
Note that this is the transpose of the monodromy
\eq{pqbrane}.
Such monodromies are
imposed since we do not want the action to vary
upon shifting the position of the branch cuts.
The only other requirement the
two-form field values themselves must
satisfy is that
\be
p_i B_{\mu \nu} + q_i C_{\mu \nu}
\ee
be well-defined at brane loci $B_i$.
Any gauge transformation with
well-defined values of \eq{gtonbloc} at $B_i$
preserve both requirements.

\begin{figure}[!t]
\centering\includegraphics[width=5cm]{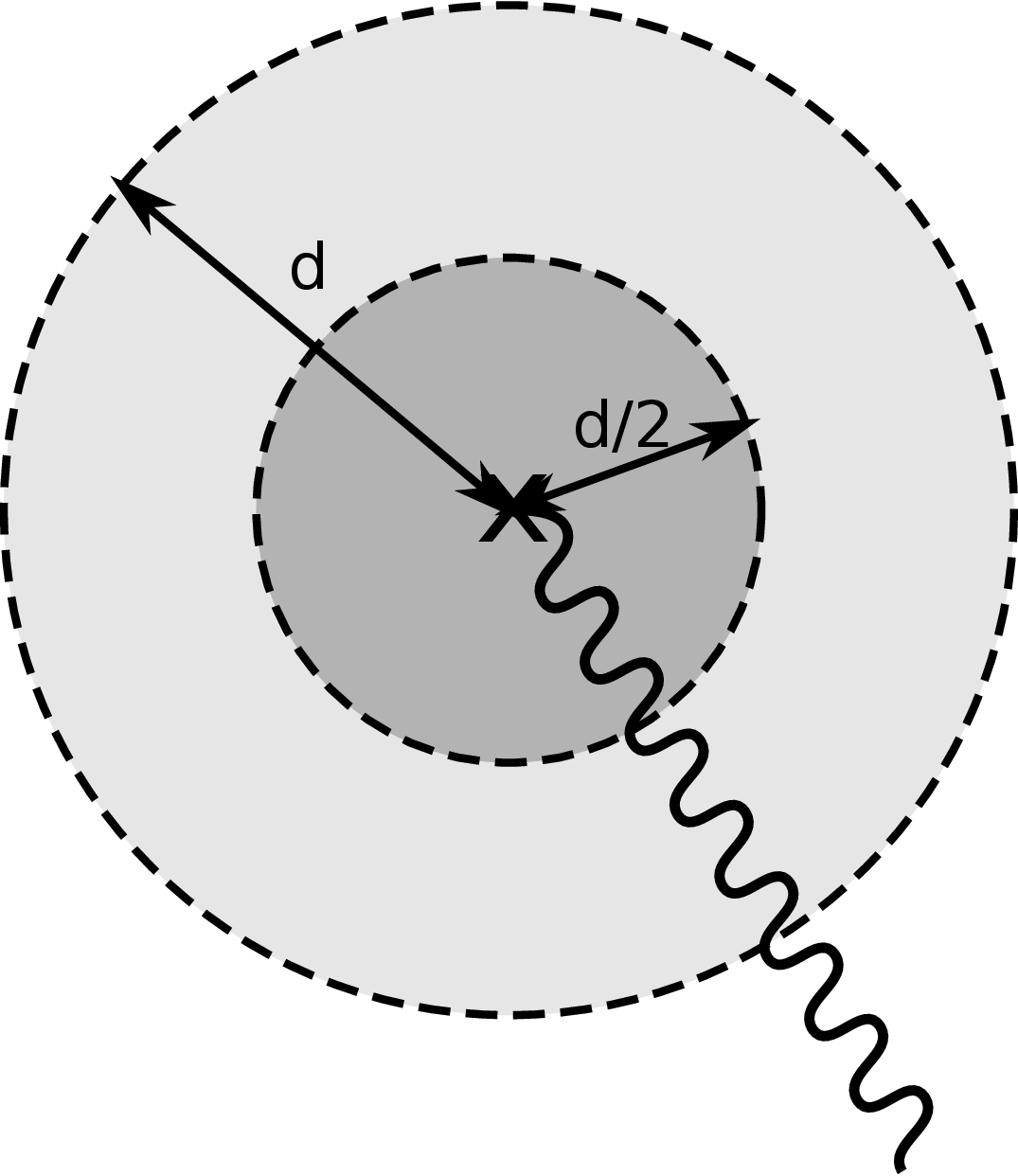}
\caption{\small A zoom-in on a small neighborhood
of a brane locus. The curvy line denotes the branch-cut
around the brane. The concentric circles of radius
$d/2$ and $d$ are chosen to be small enough so that they
do not intersect any other branch cuts.}
\label{f:contour}
\end{figure}

Now let us find a gauge transformation
that can be used to gauge away vector degrees of freedom
living on a particular brane $B_{i_0}$.
Let us assume a given seven-brane is a Dirichlet brane --- i.e.,
a $(1,0)$ brane --- and consider the gauge transformations of
the form
\be
\Gamma^I = \phi^I (z,\bar{z}) \Lambda (x^\mu)
\label{CSgt}
\ee
where $\phi^I$ is a doublet scalar dependent on the
internal coordinates while $\Lambda$ is an eight-dimensional
one-form. Taking the local holomorphic coordinate around
the brane locus $w=0$ to be $w$,
one can find some $d$ such that
the open neighborhood $|w|<d$ of the brane does not
include any other branch cut.
Then one can find a real ``bump function"
$f(w)$ that satisfies the following conditions:
\ben
\item $0 \leq f(w) \leq 1$ for all $w$, while $f(0) =1$.
\item $f(w)>0$ for $|w| < d/2$.
\item $f(w)=0$ for $|w|>3d/4$.
\item $f(w)$ is $C^\infty$ for $|w| <d$ and hence so on
the full manifold.
\een
Then, for
\be
\begin{pmatrix}
\Gamma^1 \\ \Gamma^2
\end{pmatrix}
=
\begin{pmatrix}
\phi^1 \\ \phi^2
\end{pmatrix}
\Lambda
=
\begin{pmatrix}
f(w) \\
0
\end{pmatrix} \Lambda \,,
\label{locgt}
\ee
the gauge transformation \eq{CSgt} is
well-defined on the full base manifold.
In particular, the value \eq{gtonbloc}
is well defined at each brane locus
unambiguously.
In fact, for each brane $B_i$,
\be
A_\mu^i \ra
\begin{cases}
A_\mu^i & \text{when } i\neq i_0\\
A_\mu^i-\Lambda_\mu & \text{when } i= i_0
\end{cases} \,.
\ee
The vector field $A^{i_0}$ is thereby eaten
by the two-form doublet by this CS gauge
transformation by setting $\Lambda_\mu =A^{i_0}_\mu$.
Likewise, for each brane $B_i$ of charge $(p_i,q_i)$,
we can construct a gauge transformation
\be
\begin{pmatrix}
\Gamma^1 \\ \Gamma^2
\end{pmatrix}
=
{f_i \ov p_i^2 + q_i^2}
\begin{pmatrix}
p_i \\
q_i
\end{pmatrix} \Lambda \,,
\ee
with an appropriate bump function $f_i$
to absorb $A^i$.

Hence, in ``unitary gauge" where all the gauge fields
living on the world-volume of the seven-branes are
eaten, the term \eq{kin24} of the Lagrangian becomes
\be
-{\mu_8 \ov 4} \sum_{i=1}^{24}
\int d^8 x \sqrt{-g} (p_i B_{\mu \nu}+ q_i C_{\mu \nu})^2 \,.
\label{8Dmass}
\ee
Regardless of this term, which looks like an eight-dimensional
mass term for the tensor fields, we still find modes of
the tensor fields responsible for massless degrees of freedom
of the 8D effective theory.
These come from modes with
components transverse to the seven-branes.
We proceed to examine these modes in the next section.

\section{$SL(2,\field{Z})$ Doublet Harmonic One-forms}
\label{s:doublet}

In this section, we show how 20 massless vector fields
arise in F-theory compactifications on K3 from the type IIB
perspective. In section \ref{ss:zmodes}, we
define what we mean by $SL(2,\field{Z})$ doublet
harmonic one-forms and show that the doublet
two-form of type IIB theory can be reduced along these
one-forms to yield massless vectors in the
eight-dimensional theory. In the following two sections,
we study these zero modes from two different points of view.
In section \ref{ss:h1}, we show that these harmonic forms
represent elements of the first cohomology group of a certain
sheaf living on the base manifold $\bar{S}$. We also derive
that the dimension of this cohomology group is 20.
In section \ref{ss:harm}, we establish the
correspondence between doublet harmonic one-forms
and certain closed two-forms --- namely, the ``semi-flat
harmonic two-forms" ---
living inside the underlying K3 manifold $\bar{M}$.
We end by showing how other particles
and fields of the 8D theory couple to the vectors
obtained by reducing along these one-forms
in section \ref{ss:KK}. The various couplings
are shown to encode the geometric data of $\bar{M}$.

\subsection{Definition} \label{ss:zmodes}

In this section, we show how massless vectors
arise in F-theory compactifications on K3.
These are obtained upon KK-reduction of the 
$SL(2,\field{Z})$ doublet two-forms of the type IIB theory
along doublet one-form zero modes.
In particular, we derive the ``zero mode condition" or the
``harmonicity condition" \eq{harmonic} a
$SL(2,\field{Z})$ doublet one-form $\xi^I$
living in the compactification manifold $S$
must satisfy.

Let us consider the KK-reduction of the doublet two-form fields
in the F-theory background along some doublet zero-mode.
The KK-reduction ansatz is given as
\be
B^I = \phi^I(z,\bz) \wedge b
\label{KKansatz}
\ee
where $\phi$ is some doublet $k$-form aligned along
the internal direction and $b$ is an eight-dimensional
$(2-k)$-form field.
In order for $b$ to be massless, $\phi$ must be closed:
\be
d \phi^I =0 \,.
\ee
At the same time, $\phi^I$ must respect the monodromies
defined by the background, i.e., it must exhibit the same
$SL(2,\field{Z})$ transformations that the field strengths
experience as they cross the branch cuts.
In particular, upon counter-clockwise rotation around
a brane locus $B_i$, $\phi^I$ must transform as
\be
\begin{pmatrix}
\phi^1 \\ \phi^2
\end{pmatrix}
\ra \begin{pmatrix}
1-p_i q_i &-q_i^2 \\
p_i^2 & 1+p_i q_i
\end{pmatrix}
\begin{pmatrix}
\phi^1 \\ \phi^2
\end{pmatrix}
\equiv  M_i \begin{pmatrix}
\phi^1 \\ \phi^2
\end{pmatrix}
\,.
\label{defmon}
\ee

There are two ways of seeing why such
behavior should be imposed on
$\phi^I$. The first way is by examining the
field strength of the doublet two-form
--- as explained in the previous section,
in order for the action of the
type IIB theory to be invariant under moving branch cuts,
the field strength-doublet $(dB,dC)$ must exhibit the
monodromies \eq{2fmon}. The constraint on $\phi^I$
follows by applying this condition to the ansatz
\eq{KKansatz}. Another way of seeing this constraint is
to consider the normalization of the mode $\phi^I$,
which is given by
\be
\int_{\tilde S} \MM_{IJ} \phi^I \wedge * \phi^J \,.
\label{norm}
\ee
Unless $\phi^I$ exhibits the correct monodromies,
the normalization \eq{norm} is not well defined --- in fact,
it would vary as one moves the branch cuts around.
On the other hand, when $\phi^I$ exhibit the desired
monodromies, we can define the norm
\be
\int_{\bar S}  \MM_{IJ} \phi^I \wedge * \phi^J
\equiv \int_{\tilde S}  \MM_{IJ} \phi^I \wedge * \phi^J \,.
\ee

It was pointed out in \cite{Vafa:1996xn} that there is no
way of turning on the two-form fields along the internal
directions and also satisfying the required monodromies.
Also, it is clear that there does not exist any non-zero
closed zero-forms --- i.e., constant scalars --- that exhibit
monodromic behavior. The interesting closed forms
that yield massless particles in eight-dimensions are
the one-forms, which we denote by
\be
\xi^k =
\begin{pmatrix}
\xi^{k,1} \\ \xi^{k,2}
\end{pmatrix} \,.
\ee
Hence we wish to find closed doublet one-forms
living in the base manifold $\bar{S}$.

These one-forms must be normalizable with
respect to the Hermitian inner-product
\be
\vev{\xi^k , \xi^l} \equiv
\int_{\tilde S}  \MM_{IJ} \xi^{k,I} \wedge * \xi^{l,J}
\label{hnorm}
\ee
defined by the Lagrangian, i.e.,
\be
\vev{\xi^k , \xi^k} < \infty \,.
\ee
Note that since $\xi^k$ have monodromies around the brane
loci, their components exhibit logarithmic behavior
at these points. Despite such singular behavior,
the modes can nevertheless be normalizable.
For example, when the singularities of $\xi^k$
are logarithmic, the integral near the brane loci
\be
\propto \int_0^\epsilon dr r (\ln r)^k
\ee
is convergent.

As with KK-reduction of any
$p$-form, the closed one-forms $\xi^k$ we reduce along
are defined up to an exact form $d \varphi$,
where $\varphi$ is a doublet of zero-forms.
This freedom comes as a remnant of the ``CS gauge
symmetry" of the type IIB background.
Such an ambiguity
is fixed by demanding that $\xi^k$ is harmonic
with respect to the Hermitian inner-product defined by
the Lagrangian, i.e.,
\be
d*\MM_{IJ} \xi^{k,J} =0 \,.
\label{dual}
\ee
A simple computation shows that imposing
this condition assures that $\xi^{k}$ has minimum
norm given the ``cohomology class" of $\xi^{k}$ is fixed.
In other words,
\be
\vev{\xi^k , \xi^k} < \vev{\xi^k, \xi^k} + \vev{d\varphi, d\varphi}
=\vev{\xi^k + d\varphi, \xi^k+d\varphi}
\label{harmineq}
\ee
for any non-trivial CS gauge transformation
$d\varphi \neq 0$ when
$\xi^k$ satisfies \eq{dual}.
Hence, the condition \eq{dual} singles out
elements of a certain cohomology class,
just as the usual harmonicity condition
singles out a harmonic form in a given
de Rham cohomology class.
We describe this cohomology in more
detail in the following section.

A source of worry for equation \eq{harmineq}
is that integration by parts has been used in obtaining
the equality --- for the most general CS gauge transformation
$\varphi$, the cross terms in the integral of interest may
have boundary terms.
Such boundary terms arise in the case that $\varphi$
exhibit monodromic behavior that is affine rather
than linear.
To be more precise,
let us first examine how $\varphi$ is allowed to
behave as we rotate around a brane locus.
In order for the $\xi^k+d\varphi$ to have a well defined
norm, $d \varphi$ must exhibit the monodromy
\be
\begin{pmatrix}
d\varphi^1 \\ d\varphi^2
\end{pmatrix}
\ra  \begin{pmatrix}
1-p_i q_i &-q_i^2 \\
p_i^2 & 1+p_i q_i
\end{pmatrix}
\begin{pmatrix}
d\varphi^1 \\ d\varphi^2
\end{pmatrix}
\ee
around brane $B_i$. $\varphi$ itself, however,
does not have to display this monodromy ---
in fact, it is allowed to shift:
\be
\begin{pmatrix}
\varphi^1 \\ \varphi^2
\end{pmatrix}
\ra \begin{pmatrix}
1-p_i q_i &-q_i^2 \\
p_i^2 & 1+p_i q_i
\end{pmatrix}
\begin{pmatrix}
\varphi^1 \\ \varphi^2
\end{pmatrix}
+ C_i \begin{pmatrix}
-q_i \\ p_i
\end{pmatrix} \,.
\ee
Note that the shift must be in the direction
$(-q_i,p_i)^t$, as the values
$p_i \varphi^1 + q_i \varphi^2$ must be
well-defined at the brane locus.
In the event that $C_i \neq 0$, the integration
by parts we have used in \eq{harmineq}
is no longer valid, as there would exist
boundary terms living on $\p \tilde S$
--- a contour encircling
the ``cuts" we have used to define
a type IIB frame --- that do not cancel out.

What makes the equality of \eq{harmineq}
work is that we only allow gauge transformations
such that
\be
p_i \varphi^1(z_i,\bar z_i) + q_i \varphi^2 (z_i,\bar z_i) = 0
\label{constgauge}
\ee
at seven-brane loci $(z_i,\bar z_i)$.
This is because we are computing the
eight-dimensional massless spectrum in
a ``unitary gauge" where all the seven-brane
world-volume vector fields are eaten by tensor
degrees of freedom. The condition
\eq{constgauge} ensures that the fluctuations
we consider still satisfy the unitary gauge condition.
In appendix \ref{ap:gauge}, we show that $C_i=0$
for gauge transformations $\varphi$
that do not excite seven-brane vector fields
and keep the mode $\xi^k+d\varphi$ normalizable.

Let us conclude this section by summarizing the
definition of a $SL(2,\field{Z})$ doublet harmonic
one-form $\xi$:
\[
\fbox{
\addtolength{\linewidth}{-2\fboxsep}%
\addtolength{\linewidth}{-2\fboxrule}%
\begin{minipage}{\linewidth}
\smallskip
\ben
\item $\xi$ is doublet of one-forms exhibiting
the monodromy \eq{defmon};
\be
\begin{pmatrix}
\xi^1 \\ \xi^2
\end{pmatrix}
\ra \begin{pmatrix}
1-p_i q_i &-q_i^2 \\
p_i^2 & 1+p_i q_i
\end{pmatrix}
\begin{pmatrix}
\xi^1 \\ \xi^2
\end{pmatrix}
\nonumber
\ee
around brane locus $B_i$.
\item It must satisfy the defining equations
\eq{harmonic}:
\be
d \xi^{I} =0 \,, \quad
d * \MM_{IJ }\xi^{J} =0 \,.
\nonumber
\ee
\item It must be normalizable with respect
to the inner-product \eq{hnorm}:
\be
\vev{\xi ,\eta} \equiv
\int_{\tilde S}  \MM_{IJ} \xi^{I} \wedge * \eta^{J} \,.
\nonumber
\ee
\een
\smallskip
\end{minipage}\nonumber
}
\]
In the subsequent sections, we go on to count and
construct such harmonic one-forms.

\subsection{Coordinate-free Description of Doublet
One-forms} \label{ss:h1}

In this section, we give a coordinate-free description of $\SLZ$ doublet one-forms on
$S$ in terms of certain sheaves. We also show that the dimension of the space of
doublet harmonic one-forms is equal to $20$, by applying results by Zucker \cite{Zucker}
about polarized variations of Hodge structure on curves.

Recall that $\fb \colon \Mb \to \Sb$ is an elliptically fibered K3-surface with a
section; here $\Sb$ is the complex projective line. We assume that there are exactly
$24$ singular fibers, each with a single ordinary double point --- for degree reasons,
the section cannot pass through any of the $24$ special points. If we denote by $S
\subseteq \Sb$ the complement of the $24$ singular values, and by $M = \fb^{-1}(S)$
the open surface obtained by removing the singular fibers from $\Mb$, then the
restriction $f \colon M \to S$ is a smooth family of elliptic curves. 

Sheaf theory allows us to make some of the constructions coordinate free. Instead of
choosing $A$-cycles and $B$-cycles over a dense open subset $\tilde S$ of $S$,
we can directly obtain the corresponding fiber bundle with fiber $\ZZ^2$ by defining
\[
	\shHZ = R^1 \fl \ZZ
\]
as the first higher direct image sheaf of the constant sheaf $\ZZ$ on $M$.%
\footnote{We use similar notation for other coefficient rings such as $\QQ$ or $\CC$.
Note that $\shHQ$, $\shHR$, and $\shHC$ are basically interchangeable, as they
all contain the same information.}
At each point $s \in S$, the stalk of the sheaf $\shHZ$ is equal to the first
cohomology group $H^1(E_s, \ZZ)$ of the corresponding elliptic curve $E_s =
f^{-1}(s)$. Making a choice of $A$-cycle and $B$-cycle over a simply-connected open
subset of $S$ is the same thing as choosing a local trivialization of the sheaf
$\shHZ$. 

Now the first cohomology group $H^1(E, \ZZ)$ carries a polarized Hodge
structure of weight $1$. The Hodge structure is given by the decomposition
\[
	H^1(E, \CC) \simeq H^1(E, \ZZ) \tensor_{\ZZ} \CC 
		= H^{1,0}(E) \oplus H^{0,1}(E)
\]
according to type; the polarization is given by the intersection form
\[
	Q(\alpha, \beta) = \int_E \alpha \wedge \beta.
\]
It is a polarization because of the Riemann bilinear relations: the Hermitian form
$\alpha \mapsto i^{p-q} Q(\alpha, \overline{\alpha})$ is positive-definite on the
subspace $H^{p,q}(E)$, and the above decomposition is orthogonal with respect to the
resulting Hermitian inner product on $H^1(E, \CC)$. 

Because the same is true at every point of $S$, the locally constant sheaf $\shHZ$ is
part of a polarized variation of Hodge structure $\shH$ on $S$. Let us briefly explain
what this means. The holomorphic vector bundle $\shH = \shHZ \tensor_{\ZZ} \OS$ has a
natural flat connection 
\[
	\nabla \colon \shH \to \Omega_S^1 \tensor_{\OS} \shH,
\]
with the property that the sheaf of flat sections is isomorphic to $\shHC$. Both
$\shH$ and $\nabla$ can also be constructed geometrically, and are known as a
Gauss-Manin system. The additional data coming from the polarized Hodge structures on
the cohomology of the fibers are a holomorphic subbundle $F^1 \shH$, corresponding to
the subspace $H^{1,0}(E_s)$ in the above decomposition, and a flat pairing $\shHZ
\tensor_{\ZZ} \shHZ \to \ZZ$, corresponding to the intersection pairing.

Because the locally constant sheaf $\shHZ$ contains the information about the
monodromy of $A$-cycles and $B$-cycles, and because $\shH = \shHZ \tensor_{\ZZ} \OS$,
an $\SLZ$ doublet $k$-form is easily seen to be the same thing as a smooth $k$-form
on $S$ with coefficients in the vector bundle $\shH$. We denote the space of all such
forms by the symbol
\[
	A^k(S, \shH).
\]
The connection $\nabla$ can be extended to an operator from doublet $k$-forms to
doublet $(k+1)$-forms, and the resulting complex
\[
\begin{tikzcd}[column sep=small]
0 \rar & A^0(S, \shH) \rar{\nabla} & A^1(S, \shH) \rar{\nabla} & A^2(S, \shH)
	\rar & 0
\end{tikzcd}
\]
computes the cohomology groups $H^k(S, \shHC)$ of the locally constant sheaf $\shHC$,
by a version of the Poincar\'e lemma.

One might expect naively that the number of zero-modes defined in the previous
section can be obtained by computing the dimension of the cohomology group $H^1(S,
\shHQ)$. This is not true, since elements of this cohomology group can
exhibit singular behavior near the discriminant locus, and may therefore not be
normalizable with respect to the inner-product defined in the previous section. In
fact, the dimension can be shown to be $20 + 24 = 44$.
Instead, the correct cohomology group to consider is
\[
	H^1 \bigl( \Sb, \jl \shHQ \bigr),
\]
where $j \colon S \into \Sb$ denotes the inclusion map of $S$ into $\Sb$. In
appendix \ref{ap:sheafcohom}, we use some theorems by Zucker \cite{Zucker}
to prove that $H^k \bigl( \Sb, \jl \shHQ \bigr) = 0$ for $k \neq 1$, and that
\[
	\dim H^1 \bigl( \Sb, \jl \shHQ \bigr) = 20,
\]
as expected. Moreover, it is shown in \cite{Zucker} that
the first cohomology group is isomorphic to the subspace of $A^1(S, \shH)$ consisting
of forms that are square-integrable and harmonic (with respect to the Hodge metric on
$\shH$ and the Poincar\'e metric on $S$). These two conditions are exactly the same
as in the previous section, and so we deduce that the space of $\SLZ$ doublet
harmonic one-forms is indeed $20$-dimensional.

The work of Zucker also endows $H^1 \bigl( \Sb, \jl \shHQ \bigr)$ with a polarized
Hodge structure of weight $2$. This Hodge structure is compatible with that on
$H^2(\Mb, \QQ)$; more precisely, there is a natural morphism
\[
	H^2(\Mb, \QQ)_{\perp} \to H^1 \bigl( \Sb, \jl \shHQ \bigr),
\]
and this morphism is an isomorphism of polarized Hodge structures. We will see below
how this statement about cohomology groups can be sharpened to a result about spaces
of harmonic forms.

\subsection{Construction from Harmonic Two-forms on ${M}$}
\label{ss:harm}

In this section, we explain how to construct
the doublet harmonic one-forms from the point of
view of the K3 manifold $\bar{M}$.
To be more precise, we show that
there is a correspondence between the doublet
harmonic one-forms and two-forms of
the K3 manifold that are harmonic with respect
to the semi-flat metric \cite{Greene:1989ya}.
We conjecture that these ``semi-flat harmonic
two-forms" can be obtained as a limit of
harmonic two-forms of the K3 manifold with
respect to the Calabi-Yau metric.

Let us begin by noting that a
natural way to obtain a doublet of closed
one-forms on the base manifold $S$ of the
elliptic fibration, that displays the monodromies
of the cycles of the fiber, is by using two-forms
of the underlying K3 geometry.
Consider a smooth
closed two-form $\Xi$ living inside the $M$
obtained by excising the 24 singular fibers of the
K3 manifold $\bar M$.
Now for each point $z$ on the base manifold, let us
define a doublet of forms
\be
\xi =
\begin{pmatrix}
\int_{\alpha} \Xi \\
\int_{\beta}  \Xi
\end{pmatrix}
\label{xi}
\ee
where the integration cycles are taken to lie within
the holomorphic fiber above the point $z$.
$\alpha$ and $\beta$ are the $A$ and $B$-cycle
of the fibration used to define the type IIB frame.
If $\xi$ is ``well defined,"
it is a doublet of one-forms that
exhibit the correct monodromies.
It is also closed due to the closedness of $\Xi$.

In order for $\xi$ to be well defined,
the projection of $\Xi$ to each fiber must vanish.
Meanwhile, components of $\Xi$ with both legs
parallel
to the base would not affect $\xi$ and should be
``gauged away" if one wishes to establish a
one-to-one correspondence between one-forms
on the base and two-forms in the full manifold.
Let us hence assume the components of $\Xi$
with both legs along either the fiber or the base
direction vanish.
A better presentation of this condition
is given shortly.

For $\xi$ to be harmonic
as defined in \eq{harmonic}, an additional condition
on $\Xi$ must be imposed. It is in fact that $\Xi$
should be harmonic with respect to the semi-flat
metric constructed in \cite{Greene:1989ya}.
More precisely, we consider the family of semi-flat
metrics whose K\"ahler form $J_t$ is given by
\be
J_t /i ={W(z,\bar z) \ov t} \th^z \wedge {\th^\bz} +
 {t \ov W(z, \bar z)} \th^w \wedge \th^{\bar w}
\label{semiflat}
\ee
for some function $W$ \cite{GrossWilson}.
The one forms $\th^z$ and $\th^\bz$ are
defined to as
\be
\th^z =dz, \quad \th^{\bar z} = d\bar z\,,
\ee
where $z$ is the holomorphic coordinate on the base.
$\th^w$ and $\th^{\bar w}$, which we define shortly after,
are one-forms aligned along the fiber direction.

The semi-flat metric is a Ricci-flat K\"ahler
metric on the K3 manifold, which is locally
defined by the hypersurface equation
\be
y^2 = x^3 + f_8 (z) x + g_{12} (z) \,.
\ee
$W$ is given by \eq{metric} --- we can in
fact use the Thomae formula to obtain the expression
\be
W(z,\bar z) = {4 \pi^2 \ov 2^{1/6}} {\tau_2 |\eta(\tau)|^4  \ov |\Delta|^{1/6}} 
= -i \int \lb \wedge \l \,.
\ee
$\l=dx/y$ here is the unique holomorphic one-form
of the fiber while the integral is taken to be along
the elliptic fiber at $z$.
This relation is derived in appendix \ref{ap:thomae}
--- the normalization constant is added for
aesthetic reasons.
The one-form
$\th^w$ can be expressed using the
canonical holomorphic coordinate
\be
\zeta (x)= \int^x \l
\ee
--- where the starting point of the integral is
the locus of the zero section on the fiber
--- by
\be
\th^w = d\zeta -{ (\p \l , \lb)\zeta
- (\p \l,  \l) \bar\zeta \ov (\l , \lb)} dz \,.
\ee
The one-form $\p\l$ living on the fiber can be
explicitly written as
\be
\p \l = -{f_8' (z) x + g_{12}'(z) \ov 2y^3} dx.
\ee
While $\p\l$ shows singular behavior at
certain points on the elliptic curve,
it is a ``differential of the third kind," i.e.,
its integral over closed cycles of the elliptic curve
are nevertheless well-defined.
Hence the values $(\p \l , \lb)$ and $(\p \l,  \l)$
are well-defined for the inner-product
\be
(\kappa,\l) \equiv
\int_\alpha \kappa \int_\beta \l
-\int_\beta \kappa \int_\alpha \l \,.
\ee

The one form $\th^w$ may at first sight
look rather peculiar --- it, however, can be
re-written in a simple way, using the coordinates
$x_1$ and $x_2$ which parametrize the flat elliptic fiber
such that
\be
x_1 \cong x_1 +1, \quad x_1+ix_2 \cong x_1+ix_2 + \tau(z) \,.
\ee
Upon acknowledging that
\be
(\int_\a \l) x_1 + (\int_\b \l) x_2 = \zeta \,,
\ee
it can be shown that
\be
{1 \ov W(z, \bar z)} \th^w \wedge \th^{\bar w} = {1 \ov 2\tau_2}(dx_1 +\tau dx_2)
\wedge (dx_1 +\bar \tau dx_2) \,.
\ee
We hence see that $J_t$ is aligned in the base
and fiber directions --- it is orthogonal to cycles of the manifold
that are orthogonal to the class of the fiber and the base.
Despite the nice properties of the semi-flat metric,
it fails to be a smooth Calabi-Yau metric, as it degenerates at
the discriminant locus of the elliptic fibration.
It has, however, been shown that
it is a good approximation to
the Calabi-Yau metric on the K3 manifold with
fiber size $t$ as $t$ approaches zero \cite{GrossWilson}.
It also is a smooth, non-degenerate
Calabi-Yau metric of the open manifold $M$.

In order for a doublet one-form $\xi$
constructed from a two-form $\Xi$ living inside this
manifold to be well defined,
$\Xi_{z \bar z } = \Xi _{w \bar w} =0$, as discussed
at the beginning of this section.
This condition can be expressed using the
following harmonic two-forms with respect to $J_t$:
\be
B \equiv {W(z,\bar z)} \th^z \wedge {\th^\bz},
\quad
F \equiv {1 \ov W(z, \bar z)} \th^w \wedge \th^{\bar w} \,,
\label{BandF}
\ee
as
\be
\Xi \wedge B = 
\Xi \wedge F = 0 \,.
\label{welldefined0}
\ee
Although $B$ and $F$ behave singularly
at discriminant loci, they nevertheless represent
cohomology classes of the manifold $\bar M$.
While de Rham cohomology is defined by using
smooth differential forms, a form $\Xi$ that is not
necessarily smooth still represents a cohomology
class as long as its integrals along
homology classes are well-defined.\footnote{A
representative example of
non-smooth forms with a well-defined cohomology
class is a differential form of the third kind on a
algebraic curve. Such one-forms are singular, but
have only higher order poles and no residues.
Hence the integrals of such a form along closed
cycles of an algebraic curve are well-defined.}
In this case, the cohomology class of $\Xi$ can
is given by the dual cohomology class of
\be
(\int_{C_i} \Xi) [C_i] \,,
\ee
with respect to the canonical pairing
\be
\vev{C, \om} = \int_C \om
\label{pairing}
\ee
between forms and cycles.\footnote{When we say
a cycle and a closed form are ``dual" in this paper,
we always mean that it is dual with respect to
the pairing \eq{pairing}.}
Here $C_i$ is the basis of the homology group
of the manifold.
The forms $B$ and $F$ are in fact
dual to the base and the fiber
class of $\bar M$. In particular, it can be
shown that
\be
\int_\text{Base} B =  \int_{\bar M} \Om \wedge \bar \Om,
\quad
\int_\text{Fiber} F = 1 \,,
\quad
\int_\text{Base} F=
\int_\text{Fiber}B=0 \,,
\ee
while the integrals of $B$ and $F$ over
cycles orthogonal to the base and fiber vanish.
It is worth noting that the forms $B$ and $F$
are normalizable with respect to the
semi-flat metric at finite $t$;
\be
{1 \ov t^2}\int_{\bar M} B \wedge *_{sf} B =
{t^2}\int_{\bar M} F \wedge *_{sf} F =
 \int_{\bar M} \Om \wedge \bar \Om \,.
\ee
Here, $*_{sf}$ denotes the Hodge dual
with respect to the semi-flat metric
while $\Om$ is the holomorphic two-form.

In appendix \ref{ap:dual},
we show that when $\Xi$ is harmonic
with respect to the semi-flat metric,
and its components $\Xi_{z \bar z}$ and
$\Xi_{w \bar w}$ vanish,
\be
*\MM_{IJ} \xi^J=
\begin{pmatrix}
\int_{\beta} *_{sf}\Xi \\
-\int_{\alpha}  *_{sf}\Xi
\end{pmatrix}
\label{xidual}
\ee
for $\xi$ constructed by \eq{xi}.
In this case, the condition
\be
d*\MM_{IJ} \xi^J =0
\ee
is satisfied since $*_{sf}\Xi$ is also closed.
Given this result, it is straightforward
to obtain the inner-product \eq{hnorm}
of harmonic one-form
doublets $\xi$ and $\eta$ constructed
from harmonic two-forms $\Xi$ and $H$
as an integral on $M$.
It is a simple exercise to show in fact
that
\begin{align}
\begin{split}
\vev{\xi , \eta} &\equiv
\int_{S}  \MM_{IJ} \xi^{I} \wedge * \eta^{J}
= \int_{ S} (\int_\alpha \Xi \intb *_{sf}H
-\inta *_{sf}H \intb \Xi) \\
&= \int_{M} \Xi \wedge *_{sf}H
\equiv \int_{\bar M} \Xi \wedge *_{sf}H \,,
\end{split}
\label{2formnorm}
\end{align}
using properties of harmonic forms
of the semi-flat metric derived in appendix
\ref{ap:dual}.
Hence in order for $\xi$ to be normalizable
as defined in section \ref{ss:harm}, $\Xi$
must also be normalizable with respect to
the canonical inner-product
defined by the semi-flat metric.

Before we carry further on, let us
sum up the properties of the two-forms
$\Xi$ that produce doublet harmonic
one-forms upon integrating along
cycles of the fiber:
\[
\fbox{
\addtolength{\linewidth}{-2\fboxsep}%
\addtolength{\linewidth}{-2\fboxrule}%
\begin{minipage}{\linewidth}
\smallskip
\ben
\item $\Xi$ is harmonic with respect
to the semi-flat metric, i.e., both
$\Xi$ and $*_{sf} \Xi$ are closed.
\item $\Xi$ satisfies
\be
\Xi \wedge B = 
\Xi \wedge F = 0 \,,
\label{welldefined}
\ee
for the two-forms $B$ and $F$
defined in equation \eq{BandF}.
\item $\Xi$ is normalizable with
respect to the inner-product
\be
\vev{\Xi, H}
\equiv
\int_{\bar M} \Xi \wedge *_{sf} H \,.
\label{2formnorm2}
\ee
\een
\smallskip
\end{minipage}\nonumber
}
\]
Let us denote such harmonic two-forms,
``semi-flat harmonic two-forms."
It is worth commenting that the definition
of semi-flat harmonic forms is independent
of the parameter $t$.
This is because the Hodge dual $*_{sf}$
acting on a two-form $\Xi$
is independent of $t$ when $\Xi$
satisfies \eq{welldefined}.
This, in particular, implies that these
harmonic forms can be defined at
the singular point $t =0$.

So far, we have shown that there is a
map from semi-flat harmonic two-forms of a
dense open subset $M$ of the
K3 manifold $\bar{M}$ to doublet harmonic
one-forms on the base $\bar{S}$.
We show that this map is actually
bijective in appendix \ref{ap:dual}.
This implies that the 20-dimensional
space of doublet one-forms
$H^1(\bar{S},j_* \HH_\RR)$
can be lifted to a 20-dimensional space of
semi-flat harmonic two-forms.
These two-forms are a priori defined only
on $M$.

A natural question to ask is whether these
20 semi-flat harmonic
two-forms are related to cohomology classes
of the K3 manifold $H^2 (\bar{M})$.
In appendix \ref{ap:cohomsf} we show that this
is in fact the case. To be more precise, let us first
describe the cohomology group of the open
manifold $M$. The dual homology group of $M$
is generated by 21 elements inherited from $\bar{M}$
--- the class of the section becomes trivial in $M$ ---
and 24 cycles attached to each degenerate fiber.
The 24 cycles are represented by tori $T_i$
constructed by rotating an invariant cycle
around a degeneration locus $B_i$.
In proposition \ref{prop:H2}, we prove that
the 20 semi-flat harmonic two-forms $\{ \Xi_k \}$
span the subgroup of $H^2(M)$ obtained from
pulling back $H^2(\bar{M})_\perp$ via the inclusion map
$M \into \bar{M}$.
A more practical way to say this is that
$\{ \Xi_k \}$ can be extended to the manifold $\bar{M}$,
and that it represents a basis of the cohomology group
$H^2(\bar{M})_\perp$.
Recall that $H^2(\bar{M})_\perp$
is defined to be the orthogonal space to the
cohomology classes $[B]$ and $[F]$, by which
we denote the duals of the homology classes
of the base and fiber, respectively.

Let us provide an intuitive sketch of why
the dual homology elements
of $\{ \Xi_k \}$ must lie
within the image of $H_2 (\bar{M})_\perp$
in $H_2 (M)$. This turns out to be a consequence
of imposing normalizability on $\Xi_k$.
We note that since the semi-flat harmonic two-forms
are orthogonal to the base and fiber directions,
it is enough to show that $\{ \Xi_k \}$ are orthogonal
to the cycles $T_i$ described in the preceding paragraph.
To show this, let us assume that
a semi-flat harmonic two-form $\Xi$ has a non-trivial
integral over some cycle $T_i$, i.e.,
\be
\int_{T_i}\Xi = C_i \,.
\ee
Recall that $T_i$ is constructed by rotating the
invariant cycle $(p_i \alpha+q_i \beta)$ around the
degeneration locus $B_i$.
Denoting the doublet harmonic one-form constructed
from $\Xi$ as $\xi$, this implies that
\be
\oint_c (p_i \xi^1 + q_i \xi^2) = C_i
\ee
for a contour $c$ surrounding $B_i$. Following the latter
part of appendix \ref{ap:gauge}, it can then be shown
that such $\xi$ cannot be normalizable with respect to the
norm defined for doublet one-forms, due to the divergent behavior
of $\xi$ near $B_i$. It follows that
$\Xi$ is not normalizable with respect to the
semi-flat metric, hence concluding the proof.

The fact that the semi-flat harmonic two-forms
form a basis for $H^2(\bar{M})_\perp$
suggests that they can be related to harmonic forms of
the K3 manifold with respect to a class of
smooth Calabi-Yau metrics in the following way.
Let us consider a class of smooth, non-degenerate
Ricci flat metrics with K\"ahler form $K_t$
that satisfies the following conditions:
\ben
\item The dual homology class of $K_t$
is aligned in the direction of the class of the
base and the fiber.
\item $K_t \wedge K_t = 2 \Om \wedge \bar \Om$
for the holomorphic two-form $\Om$
of the K3 manifold. In terms of the one-form $\lambda$
we have been using, this condition can be
re-expressed as
\be
K_t \wedge K_t = 2 \Om \wedge \bar \Om
= 2 {\lambda \wedge dz} \wedge
{\bar \lambda \wedge  d \bar z} \,.
\ee
\item $\int_{\bar{f}^{-1}(z)} K_t =t$ for any point
$z$ in the base, i.e., the fiber size with respect to
this metric is given by $t$.
Recall that $\bar{f}$ is the projection
map of the fibration.
\een
The semi-flat metric is also a Ricci flat
metric whose K\"ahler form $J_t$ satisfies these
conditions. While $J_t$ is degenerate at the
discriminant locus, $J_t$ and $K_t$ are
closely related --- in fact, $t J_t$ and $t K_t$
have been shown to coincide in the limit $t \ra 0$
\cite{GrossWilson}.
Since $J_t$ approximates $K_t$ well in the
small-$t$ limit, we can expect that there is
a one-to-one correspondence between
the semi-flat harmonic two forms and
harmonic two-forms of the
Calabi-Yau metric in the subspace
$H^2(\bar{M})_\perp$.
More precisely, we can put forth the
following

\medskip

\noindent\textbf{Proposal :} For the stated class
of Calabi-Yau metrics $K_t$, let
$\{ \om_1, \cdots, \om_{20} \}$ denote the linearly
independent harmonic
forms spanning the 20-dimensional subspace
$H^2(\bar{M})_\perp$ of the cohomology group
spanned by elements orthogonal to the classes
$[B]$ and $[F]$. In the limit $t \ra 0$,
these 20 harmonic forms stay linearly independent
and form a basis $\{ \Xi_1, \cdots, \Xi_{20} \}$
of the semi-flat harmonic two-forms.

\medskip

Given that the forms $\{ \om_1 ,\cdots, \om_{20} \}$
do not develop singularities that obstruct the normalizability
condition, it can be shown that $\om_k$ stay linearly
independent. This is done by examining the
duals of the cohomology classes of $\om_k$.
Since the duals of $\{ \om_k \}$
are linearly independent in the homology group,
$\{ \om_k \}$ also remain linearly independent as two-forms.

A crucial test for the validity of the proposal would
be to verify that the limits of $\om_k$
satisfy the orthogonality condition \eq{welldefined}.
This is because orthogonality at the level of
cohomology does not guarantee orthogonality
at the level of forms.
Let us consider
harmonic two-forms $\om_k$ with respect to
the Calabi-Yau metric whose cohomology class
are orthogonal to  $[B]$ and $[F]$,
i.e.,
\be
\int_{\bar{M}} \om_k \wedge \BB_t =
\int_{\bar{M}} \om_k \wedge \FF_t
= 0 \,,
\label{class}
\ee
where $\BB_t $ and $\FF_t$
are harmonic forms in the cohomology classes
$[B]$ and $[F]$.
We have added the subscripts
to emphasize that while the cohomology classes
are defined irrespective of the metric, the harmonic
forms have a metric dependence.
While the conditions \eq{class} do not imply that
the two integrands vanish at every point,
they imply ``half" of these two conditions.
Since the K\"ahler form $K_t$ of the metric
--- which is harmonic ---
is given by a linear combination
of $\BB_t$ and $\FF_t$ by assumption,
and since the Lefschetz action commutes with
the Laplacian operator, it follows that
\be
\om_k \wedge K_t
= 0 \,.
\label{orthhalf}
\ee
It would be interesting to verify that the other half
of the constraint is satisfied as the fiber size is
taken to zero.

A natural way to map the
semi-flat harmonic forms $\{\Xi_k \}$
to the corresponding harmonic forms $\{ \om_k \}$ at finite
$t$ is provided by M-theory/F-theory duality \cite{Vafa:1996xn}.
Let us consider the eight-dimensional theory obtained by
compactifying F-theory on K3 manifold $\bar M$,
and let $a_k$ be the massless vector fields
obtained by KK-reducing the doublet two-form fields
of type IIB along one-forms $\xi^k$ constructed from $\Xi_k$.
The seven-dimensional effective theory obtained
upon further compactification on a $S^1$ of radius
$\sim \alpha'/ t^{-1/2}$ --- where $\alpha'$ is the
Regge slope of the type II string ---
is dual to M-theory
compactified on $\bar M$ with Calabi-Yau metric $K_t$.
The seven-dimensional vector fields $\tilde a_k'$ --- which
are modes of $a_k$ constant along the $S^1$ ---
are obtained by KK-reducing the M-theory three-form
along harmonic forms $\om_k$.
The inverse coupling of the vector fields of the
8D theory are given by the inner-products
of the semi-flat harmonic forms $\Xi_k$ \eq{2formnorm2}.
Meanwhile, upon reduction on a circle,
the corresponding 7D couplings
receive quantum corrections from
charged particles winding around the compactification
circle. The quantum corrected
inverse couplings can be computed by the inner-product
of the harmonic forms $\om_k$ with respect to the
Calabi-Yau metric.
Hence, in this sense,
$\om_k$ can be thought of as a ``quantum corrected"
version of $\Xi_k$.

In the next subsection, we investigate various properties
of the vector fields of the 8D theory whose construction
we have been studying up to this point.
Let us conclude this section by
summarizing
what we have learned so far
about the massless vector field spectrum
of the effective 8D theory of the F-theory
compactification on $\bar M$,
and setting the conventions for the next section:
\ben
\item The massless vector spectrum comes from
reducing the type IIB doublet two-forms along
doublet harmonic one-forms. We denote
the massless vectors $a^k$ and the one-forms
$\xi^k$.
\item There are 20 linearly independent $\xi^k$.
\item The components of
$\xi^k$ can be obtained by integrating
the ``semi-flat harmonic two-forms" $\Xi_k$ of $\bar M$
along the $A$ and $B$-cycles of the elliptic fiber.
\item $\{ \Xi_1, \cdots, \Xi_{20} \}$ are closed two-forms whose
dual two-cycles span $H_2 (\bar{M})_\perp$.
\een

\subsection{Properties of KK-reduced Vector Fields} \label{ss:KK}

In this section, we explore the properties of the
massless vector fields of the 8D theory obtained by
KK-reduction. We first compute the charges of
string junctions under these vector fields. We go on to
relate the vector fields to
seven-brane world-volume vector fields through a
particular CS gauge transformation.
We conclude the section by computing
the Chern-Simons couplings of the 8D effective theory
involving the massless vector fields.

The massive charged states of the 8D theory
come from string junctions stretching inside of the
base manifold $\tilde S$ and ending on the seven-branes.
Any junction with charge vector
$\sigma$ can be represented by a tree of directed
segments $\{ s_l \}$ of $n_l(p_l, q_l)$ strings
--- where $p_l$ and $q_l$ are mutually prime ---
that either begin/end at $(p_l,q_l)$ branes
or junction points.
The segments $\{ s_{l_i} \}$
meeting at the junction points $P_i$ must satisfy the
charge conservation condition
\be
\sum_{l_i \ra P_i} n_{l_i}(p_{l_i},q_{l_i})
-\sum_{l_i \leftarrow P_i} n_{l_i}(p_{l_i},q_{l_i}) =0 \,,
\ee
where the notation $l \ra P$ ($l \leftarrow P$) is used
to denote that the segment $l$ is ending at (emanating from)
the point $P$, respectively.
The charge of any such a junction under the 8D vector
field $a^k$ is given by
\be
q_{\sigma,k} = \sum_l n_l \int_{s_l} (p_l \xi^{k,1} + q_l \xi^{k,2})
\ee
as a $(p,q)$ string couples to the doublet two-form fields
\be
\int_\Sigma (p \pi^* B+q \pi^* C)
\ee
along its world-sheet $\Sigma$. $\pi$ is the
embedding of the world-sheet in space-time.
The charge $q_{\sigma,k}$
can be expressed in terms of $\Xi_k$ as
\be
q_{\sigma,k} = \sum_l \int_{s_l} \int_{n_l p_l \alpha+
n_l q_l \beta}\Xi_k = \int_{C_\sigma} \Xi_k \,,
\label{charge}
\ee
where $C_\sigma$ is the two-cycle of the K3 manifold
that is obtained by ``fattening" the junction.
Hence the electric charge of a junction with
charge vector $\sig$ under the 8D gauge field $a^k$
is given by the topological pairing
between the cycle $C_\sig \in H_2(\bar{M})_\perp$
and the cohomology class $[\Xi_k] \in H^2 (\bar{M})_\perp$.

There is a correspondence between
these vector fields and world-volume vector
fields $A^i$ living on the seven-branes $B_i$.
This correspondence can be established by
``pushing" the eight-dimensional massless
vector modes coming from exciting the two-form
tensor $\xi^k \wedge a^k$ back into the branes via a CS
gauge transformation.
The gauge transformation we use is 
given by $\Lambda = -\varphi^k a^k$
where $\varphi^k$ is a ``doublet" one-form such that
\be
d \varphi^k =\xi^k \,.
\label{varphi}
\ee
The two-form fluctuation
\be
B = \xi^k \wedge a^k , \quad
A^i =0 \,,
\ee
is then gauge equivalent to
\be
B= -\varphi^k d a^k, \quad
A^i =(p_i\varphi^{k,1} (z_i,\bar{z}_i)
+q_i\varphi^{k,2} (z_i,\bar{z}_i)) a^k \,.
\label{akAi0}
\ee
This particular gauge transformation is implemented
so that $B$ does not have components tangent
to the internal directions, so that none of the string
junctions are charged under the $B$ components.

Now the map from $a^k$ to $A^i$ defined by \eq{akAi0}
can be expressed in terms of the topological charges
\eq{charge} either by using Stokes' theorem or
by the following observation.
Given that the fields \eq{akAi0} are turned on,
a massive state of the 8D theory
coming from quantizing a
string junction with charge vector $\sig$
is coupled to the fields via $\sig_i A^i$.
Meanwhile, this is gauge equivalent to
turning only $B=\xi^k \wedge a^k$ on.
Under this field configuration, the 8D state
is coupled to $a^k$ via $q_{\sig,k} a^k$.
Since the two field configurations are gauge
equivalent, the following identity holds:
\be
\left( \int_{C_\sig} \Xi_k \right) a^k = \sig_i A^i \,.
\label{akAi1}
\ee
This identity clearly cannot define a one-to-one
mapping between the gauge fields, as there
are $20$ of the vector fields $a^k$, while there
are $24$ seven-brane vector fields $A^i$.
There is, however, a linear subspace of all
the vector fields $A^i$ that one can identify
with the space of physical massless vector fields
of the eight-dimensional theory.

To identify this subspace, we first observe
that there is an ambiguity in defining the gauge
transformations $\varphi^k$ in \eq{varphi}
--- it is defined up to a constant. As noted in
section \ref{s:stuckelberg}, a CS gauge transformation
$\Lambda$ is allowed as long as the values
$(p_i \Lambda^1 + q_i \Lambda^2)$ at the branes
are well defined --- this allows gauge transformations
constant in the internal directions. Using this, we can
gauge away two linear combinations of $A^i$,
namely $p_i \Lambda$ and $q_i \Lambda'$
without affecting charges of string junctions.
We can therefore project away the linear
combinations $p_iA^i$ and $q_i A^i$, i.e.,
impose
\be
p_i A^i = q_i A^i =0 \,.
\ee

Next, we recall that among the
remaining $22$ linearly independent
charge vectors, there exist two charge vectors
$Z^1$ and $Z^2$
whose corresponding cycles are trivial homologically.
These two vectors are
the charge vectors of ``null junctions."
Hence, for these vectors,
\be
C_{\sigma} \cong C_{\sigma + Z^1}
 \cong C_{\sigma + Z^2} \,.
\ee
Hence for $A^i$ satisfying \eq{akAi1},
it must be that
\be
Z^1_i A^i = Z^2_i A^i = 0 \,.
\ee
Since the cohomology classes of $\Xi_k$
are linearly independent and span the full
space $H^2 (\bar{M})_\perp$, equation \eq{akAi1}
defines a bijective linear map between
the vector space spanned by $a^k$
and a subspace
\be
L = \{ c_i A^i : p_ic_i = q_i c_i = Z^1_i c_i = Z^2_i c_i = 0 \}
\ee
of the 24-dimensional space spanned by $A^i$.
We note that the usual $SO(24)$ invariant Euclidean
inner-product is used in defining this subspace, as it is
inherited from the kinetic term of the gauge fields
\eq{kin24}. This is the advertised correspondence
between bulk and brane vector fields.

Let us end the section with computing
the Chern-Simons couplings of the eight-dimensional
theory
\be
k_{lm} da^l \wedge da^m \wedge \tilde C_4 \,,
\ee
where $\tilde C$ is the 8D
four-form, which is the mode of the
type IIB self-dual four-form $C_4$
that is constant along the internal direction.
This term comes from reducing the
type IIB Chern-Simons term written in
equation \eq{IIB}.
$k_{lm}$ is then given by
\begin{align}
\begin{split}
k_{lm} &=\int_{\bar{S}} \epsilon_{IJ} \xi^{l,I} \wedge \xi^{m,J} \\
&=\int_{\bar{S}} (\int_\a \Xi_l \wedge \int_\b \Xi_m
-\int_\b \Xi_l \wedge \int_\a \Xi_m)
=\int_{\bar{M}} \Xi_l \wedge \Xi_m \,,
\end{split}
\end{align}
which is a topological intersection number of the
cohomology classes $[\Xi_l]$ and $[\Xi_m]$.
This is consistent with the picture of F-theory/M-theory
duality presented in the previous subsection.
Upon reduction along a circle of radius $r$,
one can consider the Chern-Simons coupling
\be
k_{lm} d \tilde a^l \wedge d \tilde a^m \wedge \tilde C_3 \,,
\ee
where $\tilde C_3$ is obtained by reducing
$\tilde C_4$ with one leg around the circle,
and the gauge fields $\tilde a^l$ are modes of
$a^l$ constant around the compact circle.
There are no KK-modes of the 8D fields
that correct this term in obtaining the 7D
effective action --- hence the couplings
$k_{lm}$ remain the same as in the 8D
theory.\footnote{Although absent in our case,
such corrections to topological terms must
be accounted for in general. Such issues
are addressed in
\cite{Goldstone:1981kk,D'Hoker:1984ph}.
These corrections have also recently been
discussed in the context of string compactifications
in \cite{Bonetti:2012fn,Bonetti:2012st,
Cvetic:2012xn,Bonetti:2013ela,Bonetti:2013cza}.}
This is dual to the Chern-Simons coupling of
M-theory compactified on $\bar{M}$
--- the couplings from this point of view
are given by
\be
k_{lm} =\int_{\bar{M}} \om_l \wedge \om_m \,,
\ee
where $\om_l$ are harmonic forms
with respect to the Calabi-Yau metric on
$\bar{M}$ with fiber size $\sim {\alpha'}^2/r^2$.
While the semi-flat harmonic forms $\Xi_l$
become ``quantum corrected" into harmonic
forms $\om_l$, the Chern-Simons couplings $k_{lm}$
are still given by the topological intersection numbers
of the cohomology classes $[\Xi_l]=[\om_l]$,
and thus remain the same.

\section{Future Directions} \label{s:future}

In this paper, we have examined a thoroughly
studied F-theory background through a rather
uncommon approach. Namely, we have examined
K3 compactifications directly from the point of view
of type IIB string theory, only focusing on the degrees of
freedom present there. While this approach did
not reveal anything we did not know
about K3 compactifications, we have demonstrated
that much that we know about them can be recovered
without referring to any dualities.
Hopefully, the methods employed in this paper
can be expanded to more complicated
backgrounds to address problems that are hard to
resolve by using other techniques.
Let us conclude this paper by discussing directions
in which to improve and expand our results.
\vspace{0.1in}

\noindent
{\bf Further Study of Semi-flat Harmonic Forms.}
In section \ref{ss:harm}, we have conjectured
that harmonic forms of an elliptically fibered
K3 manifold $\bar{M}$
sitting inside the subspace $H^2(\bar{M})_\perp$
of the second cohomology group behave ``nicely"
in the semi-flat limit, i.e., the limit the fiber size
of the manifold shrinks to zero.
We have further proposed that, in this limit,
they become harmonic forms with respect to the semi-flat
metric \cite{Greene:1989ya}.
Although this proposal is quite natural from the
point of view of string theory, it seems quite non-trivial
from the perspective of geometry. It would be
interesting to see if this proposal can be proved
with mathematical rigor.

Another interesting direction of research would be to
approach the semi-flat harmonic forms numerically.
The physical quantities associated to massless modes
of F-theory backgrounds are, at least classically,
computed by using semi-flat harmonic forms.
As we have demonstrated in this paper,
these harmonic forms are much simpler beasts than
the forms that are harmonic with respect to the Calabi-Yau
metric. For example, our analysis shows that the Hodge
duals and the norms of the semi-flat harmonic two-forms
are ``well-behaved" in the case of the K3 manifold.
It would be interesting to compute these quantities
explicitly using numerical methods. Hopefully,
these methods can be developed further
to apply to more complicated backgrounds,
which we now discuss.
\vspace{0.1in}

\noindent
{\bf Backgrounds
Constructed from Higher Dimensional Calabi-Yau
Manifolds.}
We have dealt with the simplest non-trivial
F-theory compactification in this paper. Generalizing our
approach to more complicated backgrounds would be interesting.
An immediate generalization would be to understand
F-theory compactified on elliptically fibered Calabi-Yau
threefolds \cite{Morrison:1996na,Morrison:1996pp}.
In this case, the base of the elliptic fibration
$f:\bar{M} \ra \bar{S}$ is two-complex-dimensional.
Although these backgrounds are very well
understood, the abelian gauge symmetry of the low-energy
effective theories of these compactifications are
rather mysterious from the point of view of type IIB string theory.

For example, let us consider a Calabi-Yau threefold
elliptically fibered over $\field{P}^2$. The low-energy effective
theory is a six-dimensional $(1,0)$ supergravity theory.
At a generic point in moduli space, the fiber has an $I_1$
singularity along a degree-$36$ curve in the base
manifold. In the type IIB picture, there is a single seven-brane
wrapping this curve. There are no vector fields in
the massless spectrum of the theory. At various points in the
complex structure moduli space, however,
the number of massless vector fields jump.
The jump can be understood when
there is enhanced non-abelian gauge symmetry ---
in this case multiple branes become coincident and the new
particles added to the massless spectrum
can be understood from the point of view
of the world-volume theory of the coincident branes
\cite{Donagi:2008ca,Beasley:2008dc,
Beasley:2008kw,Donagi:2008kj}.
The picture is less clear when only abelian
vector fields are added to the massless spectrum,
mainly due to global issues.
For example, certain loci of the Calabi-Yau moduli
space with abelian gauge symmetry can be described
by a single brane wrapping a single curve with self-crossings
in the base.\footnote{Some examples of massless abelian
vector fields showing up at such loci
can be found in \cite{Aldazabal:1996du,Klemm:1996ts}.}

The approach described in this paper seems promising
when restricted to understanding theories with
only abelian gauge symmetry
--- a naive generalization of the results of \cite{Zucker}
on computing the cohomology $H^1 (\bar{S},j_* \HH_\RR)$
seems to produce the correct number of massless vector fields.
Some subtleties, however, should be worked out.
For example, the discriminant locus of the fibration
have singular points which are not treated in \cite{Zucker}.
From the physics point of view, one must
formulate the world-volume theory
of the seven-brane wrapping the singular curve, and
also keep track of the behavior of the bulk fields and
their interactions with the world-volume
fields.\footnote{This approach was outlined in
a different context in \cite{Bershadsky:1995qy}.}
It would be interesting to achieve an understanding of
such backgrounds and compute their
massless spectrum from type IIB string theory.
Hopefully, insight gained from this process can
shed light on the more sophisticated Calabi-Yau fourfold
backgrounds relevant to F-theory
phenomenology.\footnote{Recently, there has been work
\cite{Braun:2014nva} that addresses abelian gauge symmetry
of four-dimensional F-theory backgrounds by first taking a
suitable weak-coupling (type IIB) limit of the background
\cite{Clingher:2012rg} and gradually turning on the string
coupling. It would be interesting to understand how this
approach is related to ours.}
\vspace{0.1in}

\noindent
{\bf Singular Backgrounds and Exploration
of Non-Geometric Moduli.}
Another direction to expand our results is to consider
its extension to F-theory compactifications on
singular manifolds that give rise to non-abelian
gauge symmetry. While the massless fields of
type IIB supergravity would not fully account for
the degrees of freedom of these backgrounds,
one can nevertheless ask questions about
the family of deformations of such backgrounds
to gain insight into the singular backgrounds themselves.
Observing how various modes of the
type IIB fields behave in the singular limit might
shed light on novel string backgrounds that are
difficult to probe using dualities.

In particular, the approach of
viewing F-theory backgrounds
from the type IIB perspective may
shed light on non-geometric vacua with
``F-theory fluxes." An interesting class of such vacua
can be described locally in
terms of ``T-branes" coined in \cite{Cecotti:2010bp},
and also studied in
\cite{Donagi:2003hh,Chiou:2011js,
Donagi:2011jy,Donagi:2011dv,Font:2013ida,
Anderson:2013rka}.
From the point of view of the local seven-brane,
T-branes are particular
solutions to the Hitchin-type equations
of the non-abelian fields living on coincident branes.
To be more precise, they are solutions in which the
Higgs field vacuum expectation values are
upper-diagonal. An interpretation of these configurations
in terms of the global geometry of the elliptically
fibered manifold has been given
\cite{Donagi:2011jy,Donagi:2011dv,Anderson:2013rka}.
It would be interesting to see how this geometric
data translates into the type IIB data living on the
base of the elliptic fibration.

\section*{Acknowledgements}

We thank Allan Adams, Francesco Benini,
Nikolay Bobev, Andres Collinucci,
James Gray, Jim Halverson, Chris Herzog,
Kristan Jensen, Denis Klevers, John McGreevy,
Dave Morrison, Sakura Schafer-Nameki, Wati Taylor
and Barton Zwiebach for useful discussions.
D.P. would especially like to thank Denis Klevers
for his comments on the earlier versions of
this paper,
Wati Taylor and Dave Morrison for their
support and encouragement of this work,
and the High Energy Theory Group
at the University of Pennsylvania and
the Caltech Particle Theory Group
for their hospitality while this work
was being carried out.
The work of M.D. and D.P. is supported in
part by DOE grant DE-FG02-92ER-40697.
The work of C.S. is supported by
NSF grant DMS-1331641.

\appendix

\section{Monodromies of Permitted Gauge
Transformations $\varphi$}
\label{ap:gauge}

In this appendix, we show that ``nice" CS gauge
transformations
\be
\xi^k \ra \xi^k + d \varphi
\ee
of a closed, normalizable
$SL(2,\field{Z})$ doublet one-form $\xi^k$
exhibit the monodromies
\be
\begin{pmatrix}
\varphi^1 \\ \varphi^2
\end{pmatrix}
\ra \begin{pmatrix}
1-p_i q_i &-q_i^2 \\
p_i^2 & 1+p_i q_i
\end{pmatrix}
\begin{pmatrix}
\varphi^1 \\ \varphi^2
\end{pmatrix}
\ee
around brane locus $B_i$.
We note that a general ``CS gauge transformation"
is allowed to have shifts
\be
\begin{pmatrix}
\varphi^1 \\ \varphi^2
\end{pmatrix}
\ra  \begin{pmatrix}
1-p_i q_i &-q_i^2 \\
p_i^2 & 1+p_i q_i
\end{pmatrix}
\begin{pmatrix}
\varphi^1 \\ \varphi^2
\end{pmatrix}
+
C_i \begin{pmatrix}
-q_i \\ p_i
\end{pmatrix}
\ee
as one encircles a brane.
By saying that $\varphi$ is a
``nice" gauge transformation, we mean
that it satisfies the following conditions:
\ben
\item $\xi^k + d \varphi$ is normalizable with respect
to the norm \eq{hnorm}.
\item $p_i \varphi^1 + q_i \varphi^2= 0$
at the locus of each $(p_i,q_i)$ brane $B_i$.
\een
The reason for imposing such criteria is
explained in section \ref{ss:zmodes}.

Now given that 
\be
p_i \varphi^1(z_i,\bar z_i) + q_i \varphi^2 (z_i,\bar z_i) = 0\,,
\ee
$\varphi$ near the brane locus has the form
\be
\begin{pmatrix}
\varphi^1 \\ \varphi^2
\end{pmatrix}
=
f(z,\bar z) 
\begin{pmatrix}
-q_i \\ p_i
\end{pmatrix}
+ \text{(vanishing piece as $z \ra z_i$)} \,.
\ee
Here, $f$ is a multivalued function
that exhibits the shift
\be
f \ra f + C_i
\ee
upon rotation around the brane locus.
This implies that in fact
\be
\lim_{\epsilon \ra 0} \oint_{c} df = C_i
\ee
for any contour $c$ encircling the
brane locus $(z_i,\bar z_i)$ close enough.
Hence denoting
\be
df \wedge * df = F(z,\bar z) dz \wedge d\bar z,
\ee
where $z$ is the holomorphic coordinate on the base,
$F$ has a singularity
at the brane locus that behaves like\footnote{One may
worry that the behavior $F$ of is not well-defined,
as it is rescaled upon redefinition of the
holomorphic coordinate. The behavior of the
integration measure $dz \wedge d\bar z$, however,
scales inversely, hence leaving the behavior of the
integrand invariant under coordinate redefinition.}
\be
|F| \geq {C_i^2 \ov 4 \pi^2 |z-z_i|^2} \,.
\ee
Near a given $(p_i,q_i)$ brane $B_i$, the axio-dilaton $\tau$
behaves as
\be
{r_i +s_i \tau \ov p_i+ q_i \tau} \sim {1 \ov 2\pi i} \ln (z-z_i) \,,
\ee
where $r_i$ and $s_i$ are integers such that
\be
p_i s_i -q_i r_i =1 \,.
\ee
It follows that the behavior of the integrand of the norm
\eq{hnorm} is at best given by
\be
\MM_{IJ} d\varphi^I \wedge * d\varphi^J \sim 
{C_i \ov  |z-z_i|^2 \ln |z-z_i|}  dz d\bar z 
\propto
{C_i \ov  r \ln r}  dr d\theta
\ee
near the brane locus $(z_i,\bar{z_i})$.
$(r,\theta)$ are polar coordinates centered
at the brane.
It follows that the mode
$\xi^k + d\varphi$ becomes non-normalizable
upon such a gauge transformation, unless $C_i =0$.
Therefore, in order to keep $\xi^k + d\varphi$
normalizable, all $C_i$ must vanish.
$\Box$

\section{Computation of $h^1(\bar{S},j_* \HH_C)$}
\label{ap:sheafcohom}

Here we provide proofs for the results in section \ref{ss:h1}, based on the work
of Zucker \cite{Zucker}.

\begin{proposition}\label{prop:20}
One has $\dim H^1 \bigl( \Sb, \jl \shHQ \bigr) = 20$, and all other cohomology groups
of the sheaf $\jl \shHQ$ are trivial.
\end{proposition}

\begin{proof}
According to Lemma~1.2 of \cite{CoxZucker}, we have $R^1 \fbl \QQ \simeq
\jl \shHQ$, as well as
\[
	H^0(\Sb, \jl \shHQ) = H^2(\Sb, \jl \shHQ) = 0.
\]
Now the idea is to use the Leray spectral sequence
\begin{equation} \label{eq:Leray}
	E_2^{p,q} = H^p(\Sb, R^q \fbl \QQ) \Longrightarrow H^{p+q}(\Mb, \QQ)
\end{equation}
to compare the vector space in question to the cohomology of the K3-surface. 
We observe that
\[
	R^0 \fbl \QQ \simeq R^2 \fbl \QQ \simeq \QQ,
\]
because the cohomology in degree $0$ and $2$ is one-dimensional for all fibers of
$\fb$, including the $24$ singular ones. The $E_2$-page of the spectral sequence
therefore has the following shape:
\[
\begin{tikzcd}[column sep=1ex,row sep=0]
\QQ & 0 & \QQ \\
0 & H^1(\Sb, \jl \shHQ) & 0 \\
\QQ & 0 & \QQ
\end{tikzcd}
\]
It is obvious that the spectral sequence degenerates at $E_2$; in fact,
by Corollary~15.15 of \cite{Zucker}, this is always the case.
Because we know the cohomology of a K3-surface, we get the desired result.
\end{proof}

The degeneration of the Leray spectral sequence in \eqref{eq:Leray} has the following
additional consequence.

\begin{corollary}
The Leray spectral sequence induces an isomorphism
\[
	H^2(\Mb, \QQ)_{\perp} \simeq H^1 \bigl( \Sb, \jl \shHQ \bigr),
\]
which is in fact an isomorphism of polarized Hodge structures of weight $2$.
\end{corollary}

\begin{proof}
We denote by 
\[
	L^1 H^2(\Mb, \QQ) 
		= \ker \Bigl( H^2(\Mb, \QQ) \to H^0 \bigl( \Sb, R^2 \fbl \QQ \bigr) \Bigr)
\]
the subspace of cohomology classes that restrict trivially to all the fibers of
$\fb$.  Because \eqref{eq:Leray} degenerates at $E_2$, this subspace has dimension
$21$, and the natural mapping
\[
	L^1 H^2(\Mb, \QQ) \to H^1 \bigl( \Sb, R^1 \fbl \QQ \bigr)
\]
is surjective, with kernel
\[
	L^2 H^2(\Mb, \QQ) = H^2 \bigl( \Sb, \fbl \QQ \bigr) \simeq H^2(\Sb, \QQ).
\]
As mentioned above, $H^1 \bigl( \Sb, R^1 \fbl \QQ \bigr) \simeq H^1 \bigl( \Sb, \jl
\shHQ \bigr)$; this clearly implies the desired isomorphism because
\[
	L^1 H^2(\Mb, \QQ) = H^2(\Mb, \QQ)_{\perp} \oplus L^2 H^2(\Mb, \QQ). 
\]
The isomorphism respects the polarized Hodge structures on both sides according to
the results in Section~15 of \cite{Zucker}.
\end{proof}

Let us also compute the dimension of $H^1(S, \shHQ)$. By the Picard-Lefschetz
formula, the local monodromy around each of the $24$ points of the discriminant locus
is given by the matrix
\[
	\begin{pmatrix}
		1 & 1 \\
		0 & 1
	\end{pmatrix}
\]
in a suitable basis for the first cohomology of a nearby elliptic curve. It follows
that the sheaf $R^1 \jl \shHQ$ is supported on the discriminant locus, and that its
stalk at each of the $24$ points is equal to $\QQ$. For dimension reasons, the Leray
spectral sequence
\[
	E_2^{p,q} = H^p \bigl( \Sb, R^q \jl \shHQ \bigr) 
		\Longrightarrow H^{p+q}(S, \shHQ),
\]
degenerates at $E_2$; this leads to a short exact sequence
\[
	0 \to H^1 \bigl( \Sb, \jl \shHQ \bigr) \to H^1(S, \shHQ) 
		\to \QQ^{24} \to 0.
\]
It follows that $\dim H^1(S, \shHQ) = 20 + 24 = 44$.

\section{The Thomae Formula}
\label{ap:thomae}

We write components of the semi-flat metric
in terms of the unique holomorphic one-form
\be
\lambda = {dx \ov y} .
\ee
In particular, we show that
\be
g_{z\bar z}={\tau_2 |\eta(\tau)|^4  \ov |\Delta|^{1/6}} 
= {2^{1/6} \ov 4\pi^2 i} \int \lb \wedge \l \,.
\label{thomae}
\ee
Here, $g_{z\bar z}$ is the metric of the elliptically
fibered K3 manifold
in the limit the fiber shrinks to zero-size.
As is always,
\be
\tau =\tau_1 + i \tau_2
={\int_\beta \lambda \ov \int_\alpha \lambda}
\ee
for a choice of $A$ and $B$-cycles, $\alpha$
and $\beta$.
The integral in \eq{thomae} is taken along the elliptic fiber.
Recall that the local equation defining the K3 manifold
is given by
\be
y^2 =x^3 + f_8 (z) x + g_{12} (z)
\ee
where $z$ is the base coordinate.
The discriminant $\Delta$ is defined to be
\be
\Delta \equiv 4f_8^3 + 27 g_{12}^2 \,.
\ee

Now $\tau$ satisfies the relation
\be
{(\theta_2(\tau)^8 + \theta_3(\tau)^8 + \theta_4(\tau)^8)^3 \ov \eta(\tau)^{24}}=
j(\tau) = 1728 \times {4f_8^3 \ov \Delta} \,.
\ee
Hence one finds that
\be
{|\eta(\tau)|^{4} \ov |\Delta|^{1/6}} =
6912^{-1/6} {|\theta_2(\tau)^8 + \theta_3(\tau)^8 + \theta_4(\tau)^8|^{1/2} \ov |f_8|^{1/2}}\,.
\ee
Meanwhile, by the Thomae formula,
\be
\theta_2(\tau)^8 + \theta_3(\tau)^8 + \theta_4(\tau)^8
={6 \ov \pi^4 } (\int_\alpha \l )^4 f_8 \,.
\ee
Note that both sides of this equation vary with
respect to modular transformations that come from
choosing of different
$A$ and $B$-cycles.
We therefore arrive at
\be
{|\eta(\tau)|^{4} \ov |\Delta|^{1/6}} =
32^{-1/6} \pi^{-2}
|\int_\alpha \l |^2\,.
\ee
Since
\be
\tau_2 =  {1 \ov 2i} {\int \lb \wedge \l \ov |\int_\alpha \l |^2} \,,
\ee
--- where we have used the fact that
\be
\int_\alpha \lb \int_\beta \l-\int_\alpha \l \int_\beta \lb
= \int \lb \wedge \l
\ee
--- we obtain
\be
g_{z \bar z} =
\sqrt{g} = {2^{1/6} \ov 4\pi^2 i} \int \lb \wedge \l \,.
\ee

\section{Bijection Between Semi-flat Harmonic
Two-forms of ${M}$\\
and Doublet Harmonic One-forms on $\bar{S}$}
\label{ap:dual}

In this section, we examine the properties of
semi-flat harmonic two-forms $\Xi$ living in 
${M}$ --- the open manifold obtained by excising
the degenerate fibers of the elliptically fibered K3 manifold
$\bar{M}$ --- and their integrals
along cycles of the fiber. We eventually prove that
there is a bijection between these semi-flat
harmonic two-forms defined in section \ref{ss:harm}
and the doublet harmonic one-forms
defined in section \ref{ss:zmodes}.
We use coordinate conventions
defined in section \ref{ss:harm}.

Let us begin by proving the following

\begin{proposition}
Given the one-form
$\xi$ constructed from the semi-flat
harmonic two-form $\Xi$ via \eq{xi}, i.e.,
\be
\xi =
\begin{pmatrix}
\int_{\alpha} \Xi \\
\int_{\beta}  \Xi
\end{pmatrix} \,,
\label{apxi}
\ee
the ``dual" satisfies \eq{xidual}:
\be
*\MM_{IJ} \xi^J=
\begin{pmatrix}
\int_{\beta} *_{sf}\Xi \\
-\int_{\alpha}  *_{sf}\Xi
\end{pmatrix} \,.
\ee
\end{proposition}

\begin{proof}
A closed two-form $\Xi$ whose
components $\Xi_{z \bar z}$ and
$\Xi_{w \bar w}$ vanish can be
locally written as
\be
\Xi = \om \wedge dz
+ \bar\om\wedge  d \bar z
+ f dz d \bar z \,.
\label{Xi1}
\ee
Here, $\om$ is a one-form aligned
along the fiber direction so that
\be
\om = a d\zeta + b d \bar \zeta
=a \lambda + b \bar \lambda
\ee
for some functions $a$ and $b$.
$a$, $b$ and $f$ are functions
that can {\it a priori} depend on
any coordinate. The fact
that $\Xi$ is closed imposes that
$\om$ is also closed along the fiber
direction.
Since $\l$ and $\lb$ span the space of
closed one-forms of the elliptic fiber,
$\om$ can be re-written as
\be
\om = A \lambda + B \bar \lambda
+ (\p_\zeta C d \zeta+ \p_{\bar \zeta }C d \bar \zeta)
\ee
where $C$ is a function that can depend
on the fiber coordinates $\zeta$ and $\bar \zeta$,
while $A$ and $B$ only depend on the
base coordinates $z$ and $\bar z$.
It then follows that
$\xi$ is given by
\begin{align}
\begin{split}
\xi &=
A(z,\bar z)\begin{pmatrix}
\int_{\alpha} \l \\
\int_{\beta}  \l
\end{pmatrix} dz 
+B(z,\bar z)\begin{pmatrix}
\int_{\alpha} \lb \\
\int_{\beta}  \lb
\end{pmatrix} d z \\
&+ \bar
A(z,\bar z)\begin{pmatrix}
\int_{\alpha} \lb \\
\int_{\beta} \lb
\end{pmatrix} d\bar z 
+\bar B(z,\bar z)\begin{pmatrix}
\int_{\alpha} \l \\
\int_{\beta}  \l
\end{pmatrix} d \bar z \,.
\end{split}
\label{beforelift}
\end{align}

Since 
\be
\MM_{IJ} \equiv {1 \ov \tau_2}
\begin{pmatrix}
|\tau|^2 & -\tau_1 \\
-\tau_1 & 1
\end{pmatrix}
\ee
for
\be
\tau =\tau_1 + i \tau_2
={\int_\beta \lambda \ov \int_\alpha \lambda} \,,
\ee
$\MM_{IJ}$ can be rewritten as
\be
\MM_{IJ}
={i \ov \int \lb \wedge \l}
 \begin{pmatrix}
2|\intbl|^2 & -(\intal \intblb+\intalb \intbl) \\
-(\intal \intblb+\intalb \intbl) & 2|\intal|^2
\end{pmatrix} \,.
\ee
Therefore the dual of the doublet one-form
$\xi$ as defined in \eq{dual} is given by
\begin{align}
\begin{split}
*\MM_{IJ} \xi^J=
&-A(z,\bar z)\begin{pmatrix}
\int_{\beta} \l \\
-\int_{\alpha}  \l
\end{pmatrix} dz 
+B(z,\bar z)\begin{pmatrix}
\int_{\beta} \lb \\
-\int_{\alpha}  \lb
\end{pmatrix} dz \\
&-\bar A(z,\bar z)\begin{pmatrix}
\int_{\beta} \lb \\
-\int_{\alpha}  \lb
\end{pmatrix} d\bar z 
+\bar B(z,\bar z)\begin{pmatrix}
\int_{\beta} \l \\
-\int_{\alpha}  \l
\end{pmatrix} d\bar z \,,
\end{split}
\label{dual1form}
\end{align}
for a particular normalization of
the antisymmetric two-tensor
$\epsilon_{z \bar z}$.

Let us proceed to show
that equation \eq{dual1form} can be identified with
\be
\begin{pmatrix}
\int_{\beta} *_{sf}\Xi \\
-\int_{\alpha}  *_{sf}\Xi
\end{pmatrix} \,.
\ee
Since $\Xi_{z \bar z} = \Xi_{w \bar w} =0$,
the Hodge dual $*_{sf} \Xi$ also have
vanishing components along these directions,
as the semi-flat metric factors in the base and
fiber directions.
Hence $*_{sf} \Xi$ can be
written in the form
\be
*_{sf} \Xi = \om' \wedge dz + \bar \om' \wedge d \bar z
+ f' dz d\bar z \,.
\ee
As before, since $*_{sf}\Xi$ is also closed,
$\om'$ takes the form
\be
\om' = A' \lambda +B' \bar \lambda+
(\p_\zeta C' d \zeta+ \p_{\bar \zeta }C' d \bar \zeta)
\ee
where $A'$ and $B'$ are functions independent
of the $\zeta$ and $\bar \zeta$ coordinates.
Meanwhile, applying the Hodge dual to
the expression \eq{Xi1} shows that
\be
A+ \p_\zeta C = -A'- \p_\zeta C',\quad
B+ \p_{\bar \zeta} C = B'+ \p_{\bar \zeta} C' \,,
\ee
upon a particular normalization of the antisymmetric
four-tensor
$\epsilon_{w\bar w z \bar z}$.\footnote{Such choices only
affect overall factors of the equations --- harmonicity
conditions remain the same.}
Taking partial derivatives with respect to $\bar \zeta$
and $\zeta$ on the first and second equation, respectively,
we arrive at
\be
\p_\zeta \p_{\bar \zeta} C =
\p_\zeta \p_{\bar \zeta} C' = 0 \,.
\ee
Since the only harmonic function on a compact
elliptic curve is the constant function, it must be that
\be
\p_\zeta C =  \p_{\bar \zeta} C =
\p_\zeta C' = \p_{\bar \zeta} C' = 0 \,.
\ee
It also follows that
\be
\om' = -A \lambda + B \lb \,.
\ee
Hence we arrive at
\begin{align}
\begin{split}
\begin{pmatrix}
\int_{\beta} *_{sf}\Xi \\
-\int_{\alpha}  *_{sf}\Xi
\end{pmatrix}=
&-A(z,\bar z)\begin{pmatrix}
\int_{\beta} \l \\
-\int_{\alpha}  \l
\end{pmatrix} dz 
+B(z,\bar z)\begin{pmatrix}
\int_{\beta} \lb \\
-\int_{\alpha}  \lb
\end{pmatrix} dz \\
&-\bar A(z,\bar z)\begin{pmatrix}
\int_{\beta} \lb \\
-\int_{\alpha}  \lb
\end{pmatrix} d\bar z 
+\bar B(z,\bar z)\begin{pmatrix}
\int_{\beta} \l \\
-\int_{\alpha}  \l
\end{pmatrix} d\bar z
\end{split}
\end{align}
upon integration along cycles.
This coincides with \eq{dual1form}.
\end{proof}

A corollary of this proof
is that a semi-flat harmonic two-form
takes the form
\be
\Xi = A(z,\bar z) \th^w \wedge \th^z
+B(z,\bar z) \th^{\bar w} \wedge \th^z
+\bar B(z,\bar z) \th^w \wedge \th^{ \bar z}
+\bar A(z,\bar z) \th^{\bar w} \wedge \th^{\bar z} \,,
\label{Xiform}
\ee
while the one-forms $\xi$ constructed from $\Xi$
via \eq{apxi} then take the form \eq{beforelift}, i.e.,
\begin{align}
\begin{split}
\xi &=
A(z,\bar z)\begin{pmatrix}
\int_{\alpha} \l \\
\int_{\beta}  \l
\end{pmatrix} dz 
+B(z,\bar z)\begin{pmatrix}
\int_{\alpha} \lb \\
\int_{\beta}  \lb
\end{pmatrix} d z \\
&+ \bar
A(z,\bar z)\begin{pmatrix}
\int_{\alpha} \lb \\
\int_{\beta} \lb
\end{pmatrix} d\bar z 
+\bar B(z,\bar z)\begin{pmatrix}
\int_{\alpha} \l \\
\int_{\beta}  \l
\end{pmatrix} d \bar z \,.
\end{split}
\label{xiform}
\end{align}
It is clear from these equations that
two linearly independent doublet harmonic
two-forms cannot yield the same doublet harmonic
one-form via \eq{apxi}.
Hence the bijection between doublet harmonic one-forms
and semi-flat harmonic two-forms would follow if
the following are true:
\ben
\item All doublet harmonic one-forms
are of the form \eq{xiform}.
\item Given two-forms $\Xi$ defined by
\eq{Xiform} and a doublet one-form $\xi$
defined by \eq{xiform} for the same
functions $A(z,\bar{z})$ and $B(z,\bar{z})$,
$\Xi$ is harmonic if and only if $\xi$
is harmonic.
\een

It is simple to verify that the first claim is true.
All doublet harmonic one-forms must satisfy
the appropriate monodromy conditions around
discriminant loci. The one-forms
\begin{align}
\begin{pmatrix}
\int_{\alpha} \l \\
\int_{\beta}  \l
\end{pmatrix} dz, ~~~
\begin{pmatrix}
\int_{\alpha} \lb \\
\int_{\beta}  \lb
\end{pmatrix} d z , ~~~
\begin{pmatrix}
\int_{\alpha} \lb \\
\int_{\beta} \lb
\end{pmatrix} d\bar z , ~~~
\begin{pmatrix}
\int_{\alpha} \l \\
\int_{\beta}  \l
\end{pmatrix} d \bar z \,,
\end{align}
are linearly independent
doublet one-forms that satisfy these
monodromies and are non-vanishing at
all points of $S$, i.e., points on the base
other than the discriminant loci.
We note that the components
of these one-forms are linearly
independent at each point in $S$,
as the torus is non-degenerate at
these points.
Hence any real one-form that satisfies
monodromies defined by the fibration
must be of the form \eq{xiform} for some
complex functions $A(z,\bar z)$ and
$B(z,\bar z)$.\footnote{There is a subtlety
that must be discussed here. Although
the forms $(\int_{c} \l) d\bar z $ and
$(\int_{c} \lb) d z $ --- for some cycle $c$
--- do not vanish at any point in $S$,
they are not well defined at the point that
corresponds to $z = \infty$ in the local
patch of $S$ we are working in. In fact,
taking $\tilde z = -1/z = \tilde{r} e^{i\tilde \th}$,
they behave as $\sim e^{4 i\tilde \th}d\bar{ \tilde z}$
and $\sim e^{-4 i\tilde \th}d\tilde z$, respectively.
Hence conditions on the behavior of $B$
must additionally be imposed for the doublet
one-form \eq{xiform}
to be well-behaved at this point.
Not coincidentally, this is precisely the behavior
one must impose to make $\Xi$ defined by
\eq{Xiform} to be well behaved at $\tilde z =0$,
and our following arguments go through.}

Showing the second claim is a matter
of algebra. Using the fact that
\be
d \theta^w = -{( \p \l, \bar \l) \ov (\l , \bar \l)}
\th^w \wedge \th^z
+{( \p \l,  \l) \ov (\l , \bar \l)}
\th^{\bar w} \wedge \th^z \,,
\ee
the harmonicity condition
\be
d\Xi=0, \quad d *_{sf} \Xi =0
\ee
of $\Xi$ can be shown to be equivalent to
the equations
\be
\p_{\bar z}A =0, \quad
\p_z \bar B - {( \bar\p \bar\l, \bar \l) \ov (\l , \bar \l)} B
+{( \p \l, \bar \l) \ov (\l , \bar \l)} \bar B =0 \,,
\label{harmab}
\ee
and its complex conjugates.
Meanwhile, the harmonicity conditions
\be
d\xi^I =0 ,\quad d*\MM_{IJ} \xi^J=0
\ee
are equivalent to
\be
\tilde\Lambda^I d\xi^I =0, \quad
\bar{\tilde\Lambda}^I d\xi^I =0 ,\quad
\Lambda^I d*\MM_{IJ} \xi^J=0,\quad
\bar{\Lambda}^I d*\MM_{IJ} \xi^J=0 \,,
\label{oneformharm}
\ee
for the basis of non-vanishing $SL(2,\field{Z})$
contravariant/covariant doublet scalars\footnote{These
doublets are non-vanishing at
every point of $S$ except for $z = \infty$.
Given that the behavior of $A$ and $B$ are
specified at this point, it is enough to check
the equations that they satisfy
at $S \setminus \{z = \infty \}$
for our argument to hold.}
$\tilde\Lambda$, $\bar{\tilde\Lambda}$
and $\Lambda$, $\bar\Lambda$.
$\Lambda$ and $\tilde\Lambda$
can be explicitly written as
\be
\Lambda \equiv \begin{pmatrix}
\int_{\alpha} \l \\
\int_{\beta} \l
\end{pmatrix}, \quad
\tilde\Lambda \equiv \begin{pmatrix}
\int_{\beta} \l \\
-\int_{\alpha} \l
\end{pmatrix} \,.
\ee
Plugging in the ansatz \eq{xiform} for $\xi$,
the equations \eq{oneformharm} boil down to
\eq{harmab}. We hence arrive at

\begin{corollary}
The map \eq{apxi} defines a bijection between
the space of doublet harmonic one-forms and
the space of semi-flat harmonic two-forms.
\end{corollary}

\section{The Cohomology of Semi-flat
Harmonic Two-forms}\label{ap:cohomsf}

In this appendix, we provide a coordinate-invariant description of our
method of constructing smooth two-forms on $M$ from smooth
$\shH$-valued one-forms on $S$. More generally, we shall define 
linear mappings
\[
	A^k(S, \shH) \to A^{k+1}(M, \CC)
\]
that interchange the operator $\nabla$ --- induced by the connection on the holomorphic
vector bundle $\shH$ --- and the exterior derivative $d$ on $M$.

The construction is based on the following observation. Let $X = V / \Gamma$ be a
compact complex torus of dimension $g$; here $V \simeq T_0 X$ is a $g$-dimensional
complex vector space, and $\Gamma \simeq \pi_1(X, 0)$ is a lattice of rank $2g$ in
$V$. In particular, $\Gamma \tensor_{\ZZ} \RR = V$. By the Hurewicz theorem, $H_1(X,
\ZZ) \simeq \Gamma$; now the universal coefficients theorem shows that
\[
	H^1(X, \CC) \simeq \Hom_{\ZZ}(\Gamma, \CC) \simeq \Hom_{\RR}(V, \CC).
\]
The isomorphism works like this: given an $\RR$-linear functional $\varphi \colon V
\to \CC$, we get a closed one-form $d\varphi \in A^1(V, \CC)$; it is
translation-invariant, and therefore descends to a translation-invariant closed
one-form on $X$. Note that a closed form on a compact complex torus is
translation-invariant if and only if it is harmonic for the flat metric.

Now we return to our family of elliptic curves $f \colon M \to S$. As a complex
manifold, $M$ is isomorphic to a quotient $B / \BZ$, where $p \colon B \to S$ is a
holomorphic line bundle on $S$, and $\BZ \subseteq B$ is a one-dimensional complex
submanifold that intersects every fiber of $B$ in a lattice of rank $2$. It is easy
to see that $\shH$ is isomorphic to $\OS(\Bd)$, the sheaf of holomorphic sections of
the dual bundle. According to the discussion above, a section $\sigma \in \shHC(U)$
over an open subset $U \subseteq S$ gives rise to a smooth function
\[
	\varphisig \colon p^{-1}(U) \to \CC
\]
whose restriction to every fiber of $B$ is $\RR$-linear. Its derivative $d
\varphisig$ descends to a smooth one-form on $f^{-1}(U)$ that is closed, restricts to
a translation-invariant closed one-form on every fiber of $f$, and vanishes
identically on the preferred section of $f \colon M \to S$.

More generally, let $A^k(S, \shH)$ denote the space of smooth $k$-forms on 
$S$ with coefficients in the holomorphic vector bundle $\shH$. By applying the
construction from above on a suitable open covering --- consisting of
simply connected open sets on which $\shHC$ is trivial --- we obtain linear mappings 
\[
	A^k(S, \shH) \to A^{k+1}(M, \CC),
\]
with the property that the following diagram commutes:
\[
\begin{tikzcd}
A^0(S, \shH) \rar{\nabla} \dar & A^1(S, \shH) \rar{\nabla} \dar & A^2(S, \shH) \dar \\
A^1(M, \CC) \rar{d} & A^2(M, \CC) \rar{d} & A^3(M, \CC)
\end{tikzcd}
\]
Let us consider a smooth one-form $\xi \in A^1(S, \shH)$ with $\nabla \xi = 0$; it
goes to a closed two-form $\Xi \in A^2(M, \CC)$. By construction, the restriction of
$\Xi$ to the fibers of $f$ (and the section) is identically zero, and so $\Xi \in L^1
A^2(M, \CC)$ lies in the first step of the Leray filtration. We therefore get a
well-defined linear mapping
\[
	\ell \colon H^1(S, \shH) \to L^1 H^2(M, \CC).
\]
Because $\dim S = 1$, the Leray spectral sequence for $f \colon M \to S$ degenerates
at $E_2$; it follows that the edge mapping
\[
	\eps \colon L^1 H^2(M, \CC) \to H^1(S, \shH)
\]
is an isomorphism, since $L^2 H^2(M, \CC) \simeq H^2(S, \fl \CC)$ vanishes. One
can deduce from the construction above that the composition of the edge mapping
with $\ell$ is the identity. This implies that $\ell$ is also an isomorphism. 

\begin{proposition}\label{prop:H2}
The mapping $\ell$ takes the subspace 
\[
	H^1 \bigl( \Sb, \jl \shHC \bigr) \subseteq H^1(S, \shH)
\]
isomorphically to the image of $H^2(\Mb, \CC)_{\perp}$ in $H^2(M, \CC)$.
\end{proposition}

\begin{proof}
Because all the fibers of $\fb \colon \Mb \to \Sb$ are irreducible, the long exact
sequence for the cohomology of the pair $(\Mb, M)$ shows that the kernel of the
restriction mapping
\[
	H^2(\Mb, \CC) \to H^2(M, \CC)
\]
is spanned by the class of a fiber. In particular, $H^2(\Mb, \CC)_{\perp}$ injects
into $L^1 H^2(M, \CC)$. Now the functoriality of the Leray spectral sequence gives us
the following commutative diagram:
\[
\begin{tikzcd}
H^2(\Mb, \CC)_{\perp} \rar[hook] \dar{\epsb} & L^1 H^2(M, \CC) \dar{\eps} \\
H^1 \bigl( \Sb, \jl \shHC \bigr) \rar[hook] & H^1(S, \shHC)
\end{tikzcd}
\]
It was shown in appendix \ref{ap:sheafcohom} that $\epsb$ is an isomorphism;
we also know that $\eps$ is an isomorphism. Because $\ell = \eps^{-1}$,
the assertion follows.
\end{proof}

To illustrate the meaning of this result, let us consider an $\SLZ$ doublet harmonic
one-form
\[
	\xi \in A^1(S, \shH)
\]
and the corresponding closed two-form $\Xi \in A^2(M, \CC)$. The cohomology class of
$\Xi$ lies in the image of the subspace $H^2(\Mb, \CC)_{\perp}$, but unless one knows
more about its behavior near the $24$ singular fibers, one cannot say whether $\Xi$
itself extends to a smooth closed form on $\Mb$. 
It would be interesting to understand this point better.

\end{document}